\documentclass[12pt,draftclsnofoot,onecolumn]{IEEEtran}

\usepackage[cmex10]{amsmath}
\usepackage{amssymb,setspace}
\usepackage{color,epsfig,graphics}

\usepackage{cite}
\newtheorem{theorem}{\text{Theorem}}
\newtheorem{corollary}{\text{Corollary}}
\newtheorem{lemma}{\text{Lemma}}

\newtheorem{proposition}{\text{Proposition}}
\newtheorem{remark}{\text{Remark}}
\allowdisplaybreaks
\begin{document}
\singlespacing
\title{Degrees of Freedom for the MIMO Multi-way Relay Channel \footnote{This work was supported in
    part by NSF grants 0964364 and 0964362, and is presented in part at ICCC 2012, ICC 2013 and ISIT 2013.}}
\author{\IEEEauthorblockN{Ye Tian\ \ \ Aylin
    Yener} \\
  \IEEEauthorblockA{Wireless Communications and Networking Laboratory\\
    Electrical Engineering Department\\ The Pennsylvania
    State University, University Park, PA 16802\\
    \textit{yetian@psu.edu\ \ \ yener@ee.psu.edu}}}

\maketitle
\vspace{-0.5in}

\thispagestyle{empty}
\doublespacing
\begin{abstract}
This paper investigates the degrees of freedom (DoF) of the
$L$-cluster, $K$-user MIMO multi-way relay channel, where users in each
cluster wish to exchange messages within the cluster, and they can
only communicate through the relay. A novel DoF upper bound is derived by providing users with
carefully designed genie information. Achievable DoF is identified using signal space alignment
and multiple-access transmission. For the two-cluster MIMO multi-way relay
channel with two users in each cluster, DoF is established for the general
case when users and the relay have arbitrary number of antennas, and it is shown
that the DoF upper bound can be achieved using signal space alignment
or multiple-access transmission, or a combination of both. The result
is then generalized to the three user case. For the $L$-cluster
$K$-user MIMO multi-way relay channel in the symmetric setting,
conditions under which the DoF upper bound can be achieved are
established. In addition to being shown to be
tight in a variety of scenarios of interests of the multi-way
relay channel, the newly derived upperbound also establishes the
optimality of several previously established achievable DoF results for multiuser relay channels that are
special cases of the multi-way relay channel.
\end{abstract}

\newpage
\setcounter{page}{1}
\section{Introduction}
The multi-way relay channel \cite{Multiway} is a
fundamental building block for relay networks with multicast
transmission, and can model several interesting communication scenarios. In cellular
networks, a set of mobile users can form a social network by forming clusters and
exchange information by communicating via the base station, which
serves as the relay in the multi-way relay channel. In ad hoc
networks, wireless nodes can be geographically
separated, yet they can communicate to a central controller to
share information in groups. This model is also relevant to satellite
communications, where the satellite serves as the relay and the users
have multicast information that needs to be shared with the help of
the satellite \cite{Multiway}.

The simplest special case of the multi-way relay channel is the two-way relay
channel, which consists of two users that wish to exchange
information with the help of a relay. The capacity of the two-way
relay channel has been studied extensively, see for example
\cite{shannontwoway,Aves_twoway,twowayconst,bidirection,MinChenMUTWRC} and
the references therein. Even for this simplest set up,
only constant gap capacity results is known \cite{twowayconst}, achieved by physical layer
network coding, or functional decode-and-forward (FDF). 

In its general form, the multi-way relay channel, contains an arbitrary
number of clusters containing arbitrary number of users that want to exchange information. The relay needs to
handle interference that results from simultaneous transmissions of
different clusters, and the users need to recover the intended
messages in the presence of interfering signals containing messages for other
users. One might expect the strategies designed for the
two-way relay channel to be helpful, but more sophisticated
strategies are needed to handle the co-existence of messages intended for
different users.

The exact capacity characterization for the multi-way relay channel has been considered in references 
\cite{Multiway,Ong,Equalrate_twoway,LinearY,Ysum,divide}. Specifically,
reference \cite{Multiway} has proposed the general
multi-way relay channel model, and
characterized the upperbounds on the capacity region and established
achievable rates based on decode-and-forward (DF),
compress-and-forward (CF), amplify-and-forward (AF), and using nested
lattice codes. Reference \cite{Ong} has considered the special case when there is
one cluster of users, and each user wishes to exchange information
with the rest of the users. The capacity region is characterized for a finite field channel. It is shown that, for this case,
functional-decode-forward (FDF) combined with rate splitting and joint
source-channel coding achieves capacity. For the Gaussian multi-way
relay channel with one cluster, capacity result is obtained for some
special cases when the channel is symmetric using FDF \cite{Equalrate_twoway}. For the asymmetric multi-way
relay channel with a single cluster, also known as the Y channel, references \cite{LinearY,Ysum}
have obtained a constant gap capacity result
for all channel coefficient values for the three-user case. Reference \cite{divide} has studied the
multi-way relay channel with two clusters and each cluster has two
users with a single antenna, and established a
constant gap capacity result using a combination of lattice codes and
Gaussian codes. For the multi-cluster set up with two users in each cluster, i.e., the multi-pair two-way relay channel, reference
\cite{minchen_mutwoway} has studied the detection and interference
management strategies, and reference \cite{MinChenMUTWRC} has studied power allocation with orthogonal channels. For the general multi-way relay channel,
the exact capacity remains unknown due to the complexity of channel,
in turn making it difficult to obtain design insights for the
general setting. The DoF characterization, which studies how reliable communication rate scale with power, on the other hand, can provide us insights
about the optimal signal interaction in time/frequency/space
dimensions, and can be useful to design transmission schemes to achieve
higher rates.

The DoF of a wireless system characterizes its high
signal-to-noise ratio (SNR) performance. Interference alignment \cite{DoFMIMOX} has been shown to achieve
the optimal DoF for various wireless multi-terminal network models
\cite{CadambeKuser,DoFMIMOX,CadambeXNetwork,TiangaoKMN}. The essence of
interference alignment lies in keeping the interference signals in the
smallest number of time/frequency/space dimensions, and enabling the
maximum number of independent data streams to be transmitted. A similar concept, signal space
alignment, which is a special form of FDF, is proposed in reference \cite{signalspaceY} for the $Y$
channel. In this reference, the authors have shown that, by aligning the signals from
the users that want to exchange information at the same dimension,
network coding can be utilized to maximize the utilization of the
spatial dimension available at the relay to achieve the optimal DoF. In essence,
the goal of signal space alignment is to align the useful signals
together to maximize the utilization of signal dimension, whereas
the goal of interference alignment is to align
the harmful signals together to minimize the effect of
interference. In reference \cite{KYchannel}, the signal space
alignment idea is extended to the $K$-user Y channel, which has $K$
users in a single cluster that want to exchange information, and the
achievable DoF is established. For the MIMO multi-pair two-way relay channel, reference
\cite{IAinMUTWRC} has studied the requirement for the number of
antennas at the users to allow them exchange information with
the help of the relay without interfering each other in the symmetric
setting. The DoF of this channel is further studied in
reference \cite{DoF-MIMO_ML_TWRC}. Signal space alignment is further utilized
in reference \cite{asymmw}, which has considered a different variation of the
MIMO multi-way relay channel where a base station wants to exchange
information with $K$ users with the help of a relay. The DoF
of this model is established under some specific relations between the
number of antennas at the relay and at the users.

For the DoF characterization of MIMO multi-way relay channels, the known DoF
upperbound obtained to date is a cut-set bound, which can provide a tight upperbound
for the two-way relay channel, three-user $Y$ channel and two-cluster multiway relay channel
with two users in each cluster, but can be arbitrarily
loose for other instances of the model.

In this work, we derive a new DoF upper bound for the $L$-cluster
$K$-user MIMO multi-way relay channel using a genie-aided approach, such that the user with enhanced signal and
a carefully designed set of genie information can decode a subset of
messages from the other users. We show that the DoF for the MIMO multi-way relay
channel is always upper bounded by $2N$ with $N$ being the number of
antennas at the relay. This DoF upper bound, combined with the cut-set
bound, provides us a comprehensive set of DoF upper bounds for the
general MIMO multi-way relay channel. This allows us to show that the DoF upperbound
is tight for some achievable DoF results in previous works
corresponding to special cases of the multi-way relay channel.

Next, we investigate
the achievable DoF for several scenarios of the
MIMO multi-way relay channel. We utilize signal space alignment, where the users utilize the signal
space of the relay in common, and the relay can decode a function of
the transmitted signals from a pair of users and multiple-access transmission, where
the users do not share the signal space of the relay,
and the relay simply decodes the transmitted signals as in the multiple
access channel to establish the achievability results.

For clarity, we first consider the case with two clusters
and each cluster has two users, and the users and the relay can have
arbitrary number of antennas. We show that for some cases, signal
space alignment achieves the optimal DoF. For the remainder of the
cases, the DoF upper bound can be achieved using multiple-access
transmission or a combination of multiple-access transmission and
signal space alignment. The DoF
results imply that letting users utilize the relay in common is not
always optimal. Additionally, for some cases using only a subset of antennas at the relay can achieve
higher DoF. We next generalize the results to the case
with two clusters and each cluster has three users, and obtain the optimal
DoF for several scenarios of interests. We then consider the
$L$-cluster $K$-user MIMO multi-way relay channel in the symmetric setting,
where all users have the same number of antennas. Conditions between
the number of antennas at the relay and the users are established when
the DoF upper bound can be achieved. The DoF result implies that the DoF for the
MIMO multi-way relay channel is always limited by the spatial dimension
available at the relay, and increasing the number of users and clusters
cannot achieve DoF gain when the number of antennas at the relay
is limited. Furthermore, since using signal
space alignment to share the signal space of the relay between two
users can provide 1 bit for 2 bits gain, the DoF upper bound $2N$
provides the insight that we cannot obtain any further DoF gain by letting three or more users to share
the same spatial dimension of the relay.

The remainder of the paper is organized as follows. Section
\ref{sec:system-model} describes the channel model. Section
\ref{sec:dof-upperbound} derives the new DoF upper bound for the
general MIMO multi-way relay channel. Section
\ref{sec:two-cluster-multiway} investigates the DoF for the
two-cluster MIMO multi-way relay channel. Section
\ref{sec:general-multiway} investigates the DoF for the general
MIMO multi-way relay channel. Section \ref{sec:conclusion} concludes the
paper.

\section{Channel Model}\label{sec:system-model}
\begin{figure}[t]
\centering
\includegraphics[width=4in]{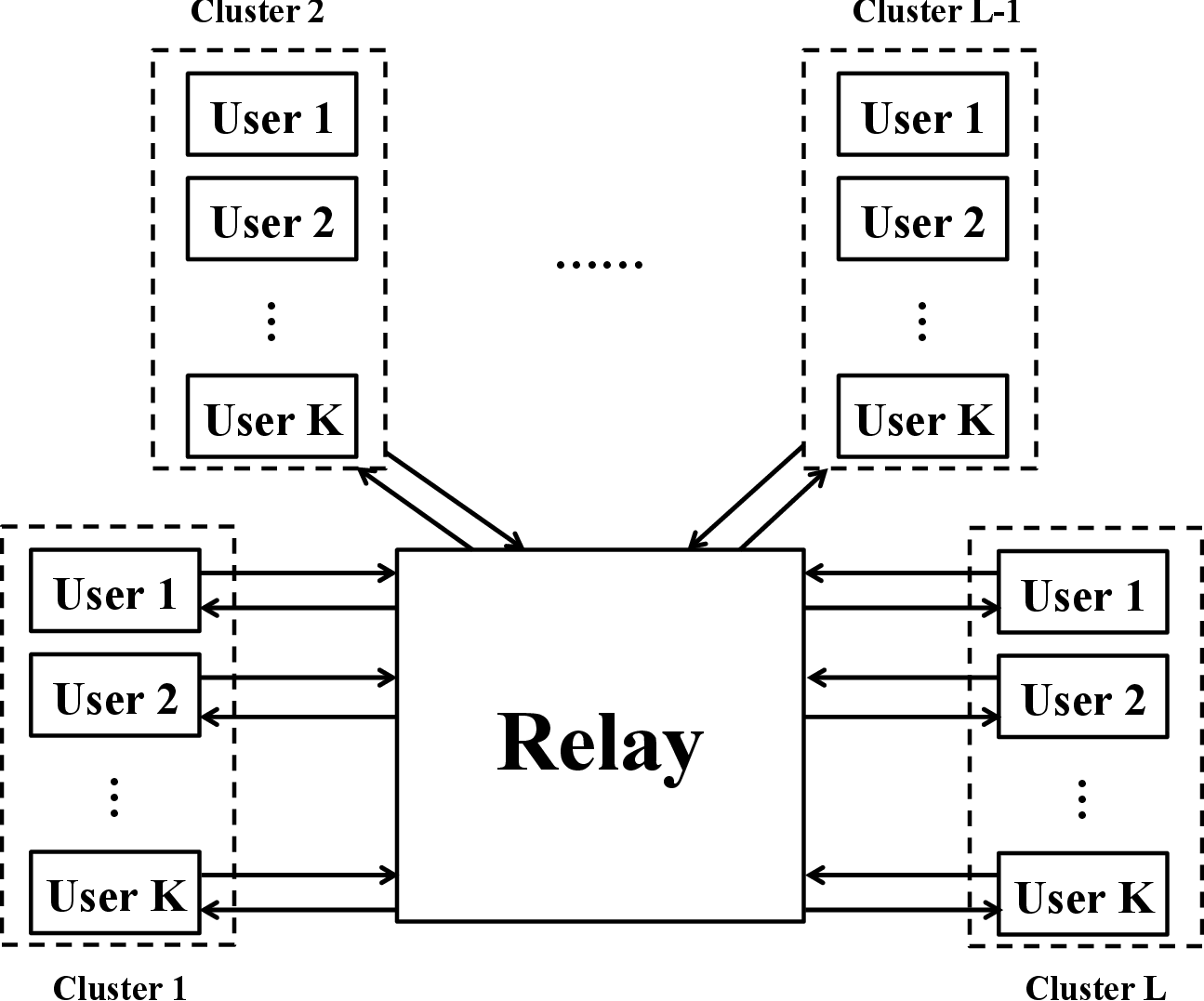}
\caption{K-user L-cluster MIMO multi-way relay channel.} \label{fig:system}
\end{figure}

The $L$-cluster $K$-user MIMO multi-way relay channel is shown in
Fig. \ref{fig:system}. User $k$ $(k=1,2,\cdots,K)$ in cluster $l$
$(l=1,2,\cdots,L)$ is assumed to have $M_k^l$ antennas, and the relay
is assumed to have $N$ antennas. Without loss of generality, we assume
that $M_1^l\ge M_2^l \ge \cdots \ge M_K^l$.

In cluster $l$, user $k$ has a message $W_{ik}^l$ $(i=1,2,\cdots,K,ï½ž
i\ne k)$, for all the other users in
cluster $l$. We denote
$\mathcal{W}_k^l$ as the message set originated from user $k$ in
cluster $l$ for all the other users in the same cluster, i.e.,
\begin{equation}
\mathcal{W}_k^l=\{W_{1k}^l,W_{2k}^l,\cdots,W_{k-1,k}^l,W_{k+1,k}^l,\cdots,W_{Kk}^l\}.
\end{equation}

It is assumed that the users can communicate only through the
relay and no direct links exist between any
pairs of users \cite{Multiway}. All the nodes in the network are assumed to be full duplex. The
transmitted signal from user $k$ in cluster $l$ for channel use $t$ is denoted as
${\bf X}_{k,l}(t)\in \mathbb{C}^{M_k^l}$. The received signal at the
relay for channel use $t$ is denoted as ${\bf Y}_R(t)\in
\mathbb{C}^N$. The received signal at user $k$ in cluster $l$ for
channel use $t$ is defined as ${\bf Y}_{k,l}(t)\in
\mathbb{C}^{M_k^l}$. The channel matrix from user $k$ in cluster $l$
to the relay is denoted as ${\bf H}_{R(k,l)}(t)\in \mathbb{C}^{N\times
  M_k^l}$. The channel matrix from the relay to user $k$ in cluster
$l$ is denoted as ${\bf H}_{(k,l)R}(t)\in \mathbb{C}^{M_k^l\times
  N}$. It is assumed that the entries of the channel
matrices are drawn independently from a continuous distribution, which
guarantees that the channel matrices are full rank almost surely.

The encoding function at user $k$ in cluster $l$ is defined as
\begin{equation}
{\bf X}_{k,l}(t)=f_{k,l}(\mathcal{W}_k^l,{\bf Y}_{k,l}^{t-1}),
\end{equation} where ${\bf Y}_{k,l}^{t-1}=[{\bf Y}_{k,l}(1),\cdots,{\bf Y}_{k,l}(t-1)]$.

The received signal at the relay is
\begin{equation}
{\bf Y}_R(t)=\sum_{l=1}^L\sum_{k=1}^{K}{\bf H}_{R(k,l)}(t){\bf X}_{k,l}(t)+{\bf Z}_R(t).
\end{equation}

For channel use $t$, the transmitted signal ${\bf X}_{R}(t)\in \mathbb{C}^N$ from the relay is a function of
its received signals from channel use $1$ to $t-1$, i.e.,
\begin{equation}
{\bf X}_R(t)=f_R({\bf Y}_R^{t-1}).
\end{equation}

The received signal at user $k$ in cluster $l$ for channel use $t$ is
\begin{equation}
{\bf Y}_{k,l}(t)={\bf H}_{(k,l)R}(t){\bf X}_R(t)+{\bf Z}_{k,l}(t).
\end{equation}

In the above expressions, ${\bf Z}_{k,l}(t) \in \mathbb{C}^{M_k^l},{\bf
  Z}_R(t)\in \mathbb{C}^N$ are additive white Gaussian noise
vectors with zero mean and independent components. The
transmitted signals from the users and the relay satisfy the following
power constraints:
\begin{equation}
E\left[{\rm tr}\left({\bf X}_{k,l}(t){\bf X}_{k,l}(t)^\dagger\right)\right]\le
    P,
\end{equation}
\begin{equation}
E\left[{\rm tr}\left({\bf X}_R(t){\bf X}_R(t)^\dagger\right)\right]\le P.
\end{equation}

Based on the received signals and the message set $\mathcal{W}_k^l$,
user $k$ in cluster $l$ needs to decode all the messages intended for
it, which is denoted as
\begin{equation}
\hat{\mathcal{W}}_k^l = \{\hat{W}_{k,1}^l,\hat{W}_{k,2}^l,
\cdots,\hat{W}_{k,k-1}^l,\hat{W}_{k,k+1}^l,\cdots,\hat{W}_{k,K}^l\}.
\end{equation}
We also have
\begin{equation}
\hat{\mathcal{W}}_k^l=g_{k,l}({\bf Y}_{k,l}^n,\mathcal{W}_{k}^l)
\end{equation} where $g_{k,l}$ is the decoding function for user $k$
in cluster $l$.

We assume the rate of message $W_{ik}^l$ is $R_{ik}^l(P)$ under power
constraint $P$. A rate tuple $\{R_{ik}^l(P)\}$ with $l=1,\cdots,L$, $k=1,\cdots,K$ and
$i=1,\cdots,K$, $i\ne k$ is achievable if the
error probability
\begin{equation}
P_e^n={\rm Pr}\left(\bigcup_{l,k,i}\hat{W}_{i,k}^l\ne W_{i,k}^l\right) \to 0
\end{equation} as $n\to \infty$.

We define $\mathcal{C}(P)$ as the set of all
achievable rate tuples $\left\{R_{ik}^l(P)\right\}$, under
power constraint $P$. The degrees of freedom is defined as
\begin{equation}
DoF = \lim_{P\to \infty} \frac{R_{\sum}(P)}{\log (P)},
\end{equation} where
\begin{align}
R_{\sum}(P)=&\sup_{{\left\{R_{ik}^l(P)\right\}\in\mathcal{C}(P)}}
\sum_{l=1}^L\sum_{k=1}^K\sum_{\substack{i=1\\i\ne k}}^KR_{ik}^l(P)
\end{align}
is the sum capacity under power constraint $P$.

\section{DoF Upperbound for General MIMO multi-way Relay Channel}\label{sec:dof-upperbound}
\begin{theorem}\label{thm-DoF-upperbound}
For the general $L$-cluster $K$-user MIMO multi-way relay channel,
the DoF upperbound is
\begin{equation}
DoF\le \min\left\{\sum_{l=1}^L\sum_{k=1}^KM_k^l,2\sum_{l=1}^L\sum_{k=2}^KM_k^l,2N\right\}.\label{eq:4}
\end{equation}
\end{theorem}
\begin{proof}
The first two terms of the upperbound can be derived using a cut set
bound as follows. Note that by assumption we have $M_1^l\ge M_2^l \ge \cdots \ge
M_K^l$. For the messages in cluster $l$, we give the users in cluster
$l$ all the messages from all the other clusters. We also provide the
relay all the messages from all the clusters except cluster $l$, and
provide all the other clusters all the messages from all clusters. This operation
does not reduce the rate of the messages in cluster $l$. Now the
channel is effectively a MIMO multi-way relay channel with a single
cluster. For the messages intended for user $i$, we can combine all
the other users except for user $i$, which yields a two-way relay
channel with user $i$ as a node, and all the other users as a node. We can then
bound the DoF for the messages in cluster $l$ in the following fashion:
\begin{align}
&\sum_{k=2}^Kd_{1k}^l\le \min\{M_1^l,N,\sum_{k=2}^KM_k^l\}\\
&\sum_{\substack{k=1\\k\neq i,i\neq 1}}^Kd_{ik}^l\le \min\{M_i^l,N\}.
\end{align} This yields the desired DoF upperbound for the first two terms in
(\ref{eq:4}).

To prove $DoF\le 2N$, we consider user $1$ in each cluster. For user $1$ in cluster $l$, we provide user $1$ with all the received signals ${\bf Y}_{k,l}^n$ from user $k=2,3,\cdots,K$ in cluster $l$ as genie information \footnote{Alternatively, we can use a channel enhancement argument to prove the DoF upperbound, as shown in \cite{YeISIT13}.}. By assumption, user $k$ in cluster $l$ can decode messages $\left\{W_{ki}^l\right\}$ for
$i=\{1,\cdots,K\}/\{k\}$ given the received signal ${\bf Y}_{k,l}^n$ and
the side information $\mathcal{W}_{k}^l$. We now use the following steps to derive the
DoF upper bound.

\begin{itemize}
\item Step 1:
\begin{itemize}
\item User 1 in cluster $l$ has side information $\mathcal{W}_{j1}^l$,
$j=\{2,\cdots,K\}$, which are messages that originate from it.
\item User 1 in cluster $l$ can decode messages $\left\{W_{1i}^l\right\}$ for
$i=\{2,\cdots,K\}$.
\end{itemize}
\item Step 2:
\begin{itemize}
\item Let a genie provide user $1$ in cluster $l$ with messages
$\left\{W_{i2}\right\}$, $i=\{3,\cdots,K\}$.
\item User 1 now has the messages
  $\left\{W_{i2}\right\}$, $i=\{1,3,\cdots,K\}$, which is exactly
  all the side information available at user 2 in cluster $l$.
\item With the genie information ${\bf Y}_{2,l}^n$, user 1 in cluster
  $l$ can decode all the messages intended for user 2 in cluster $l$,
  i.e., the messages $\left\{W_{2i}^l\right\}$ for $i=\{3,\cdots,K\}$.
\end{itemize}
\item Step 3:
\begin{itemize}
\item Let a genie provide user 1 in cluster $l$ with messages
  $\left\{W_{i3}\right\}$, $i=\{4,\cdots,K\}$.
\item User 1 now has the messages
  $\left\{W_{i3}\right\}$, $i=\{1,2,4,\cdots,K\}$, which is exactly
  all the side information available at user 3 in cluster $l$.
\item With the genie information ${\bf Y}_{3,l}^n$, user 1 in cluster $l$ can
  decode the messages $\left\{W_{3i}^l\right\}$ for $i=\{4,\cdots,K\}$.
\end{itemize}
\item Proceed in the same fashion for the next steps, i.e., for Step $k$:
\begin{itemize}
\item Let a genie provide
  user 1 in cluster $l$ with messages $\left\{W_{ik}\right\}$,
  $i=\{k+1,\cdots,K\}$.
\item User 1 now has the messages
  $\left\{W_{ik}\right\}$, $i=\{1,\cdots,K\}/\{k\}$, which is exactly
  all the side information available at user $k$ in cluster $l$.
\item With the genie information ${\bf Y}_{k,l}^n$, user 1 in cluster $l$ can
  decode the messages $\left\{W_{ki}^l\right\}$ for $i=\{k+1,\cdots,K\}$.
\end{itemize}
\item Step $K-1$:
\begin{itemize}
\item Let a genie provide user 1 in cluster $l$ with message
  $W_{K,K-1}^l$.
\item User 1 now has the messages $\left\{W_{i,K-1}\right\}$,
  $i=\{1,\cdots,K-2,K\}$, which is exactly all the side information
  available at user $K-1$ in cluster $l$.
\item User 1 in cluster $l$ can decode the message
  $W_{K-1,K}$.
\end{itemize}
\end{itemize}
\begin{figure}[t]
\centering
\includegraphics[width=5in]{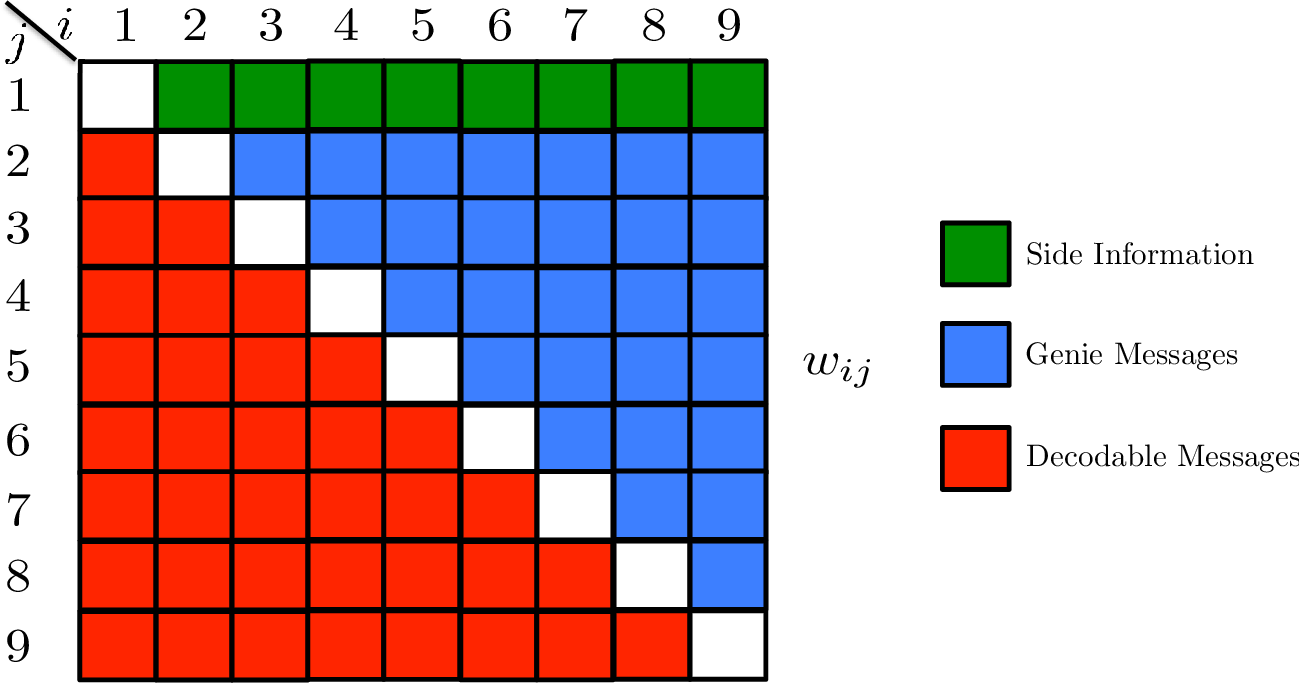}
\caption{Illustration for side information, genie information and
  decodable messages for the DoF upperbound at user $1$.} \label{fig-outerbnd}
\end{figure}

\begin{figure}[t]
\centering
\includegraphics[width=5in]{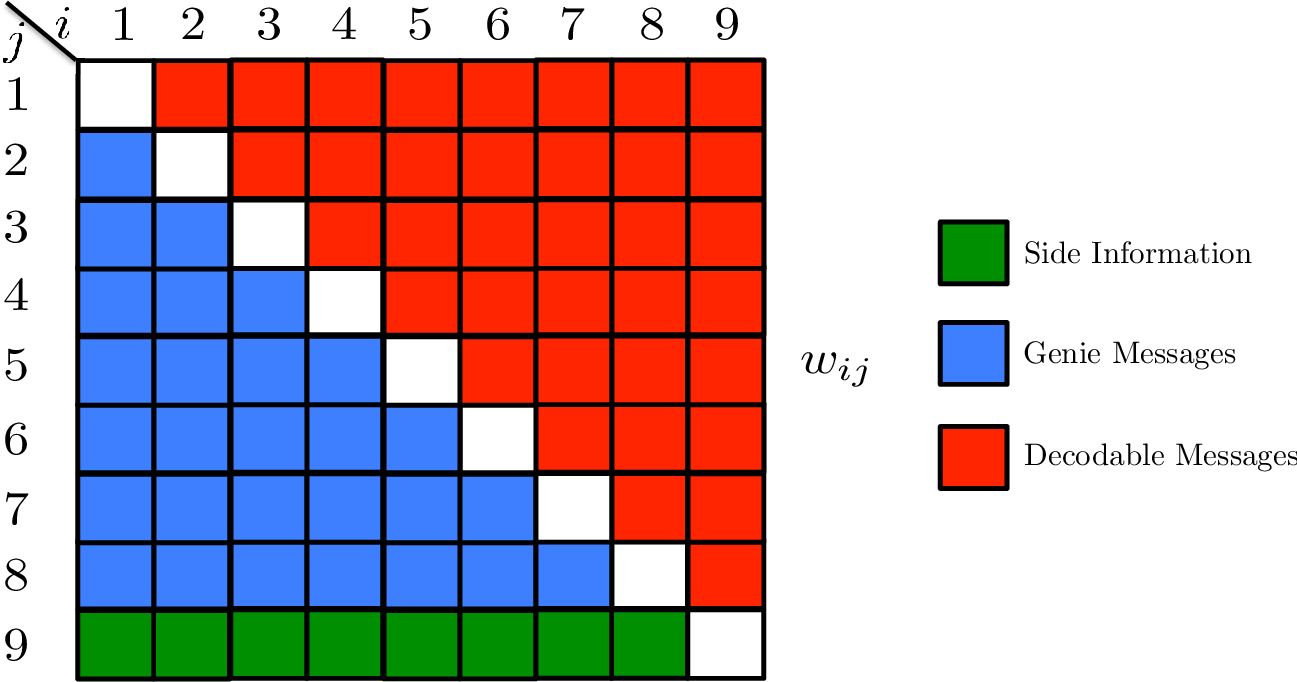}
\caption{Illustration for side information, genie information and
  decodable messages for the DoF upperbound at user $K$.} \label{fig-outerbnd_2}
\end{figure}
Based on the above arguments, user 1 in cluster $l$ can decode the
messages $\left\{W_{ki}^l\right\}$ for $k=1,\cdots, K-1$,
$i=k+1,\cdots,K$. This upperbounding process is illustrated in
Fig. \ref{fig-outerbnd} for $K=9$. We can see that half of the messages in cluster $l$ can
be decoded at user 1 in cluster $l$ based on the received signals
${\bf Y}_{k,l}^n$, $k=1,\cdots,K$, and the other half of the messages in the cluster as side information, which include the messages $\mathcal{W}_{1}^l$ and the genie information $\left\{W_{ik}^l\right\}$ for $k=2,\cdots,K-1$, $i=k+1,\cdots,K$.

Define $\mathcal{W}_{d}^l$ as the set of messages $\left\{W_{ki}^l\right\}$ for $k=1,\cdots, K-1$,
$i=\{k+1,\cdots,K\}$ for cluster $l$, which are messages that can be
decoded by user 1 in cluster $l$, and $\mathcal{W}^l$ as all the messages from
cluster $l$ and $\mathcal{W}_d^{lc}$ as the set of messages
$\mathcal{W}^l/\mathcal{W}_d^l$. Denote the set of received signals ${\bf Y}_{k,l}$, $k=1,\cdots,K$ in cluster $l$ by $\mathcal{Y}_{l}$. We can then bound the rate of the
decodable messages as follows:

\begin{align}
&n\sum_{l=1}^L\sum_{k=1}^{K-1}\sum_{i=k+1}^{K}R_{ki}^l\\
&=H(\mathcal{W}_d^1,\cdots,\mathcal{W}_d^L|\mathcal{W}_d^{1c},\cdots,\mathcal{W}_d^{Lc})\\
&=I(\mathcal{W}_d^1,\cdots,\mathcal{W}_d^L;\mathcal{Y}_{1}^n,\cdots,\mathcal{Y}_{L}^n|\mathcal{W}_d^{1c},\cdots,\mathcal{W}_d^{Lc})+n\epsilon_n\\
&\le
  H(\mathcal{Y}_{1}^{n},\cdots,\mathcal{Y}_{L}^{n})-H(\mathcal{Y}_{1}^{n},\cdots,\mathcal{Y}_{L}^{n}|\mathcal{W}^{1},\cdots,\mathcal{W}^{L},X_R^n)+n\epsilon_n\\
&=H(\mathcal{Y}_{1}^{n},\cdots,\mathcal{Y}_{L}^{n})-H(\mathcal{Y}_{1}^{n},\cdots,\mathcal{Y}_{L}^{n}|X_R^n)+n\epsilon_n\label{eq:11}\\
&=I(X_R^n;\mathcal{Y}_{1}^{n},\cdots,\mathcal{Y}_{L}^{n})\label{eq:1}
\end{align} where equation (\ref{eq:11}) follows since the received
signal at the users only depends on the transmitted signal from the relay.

From equation (\ref{eq:1}), we can see that
\begin{equation}
\lim_{SNR\to
  \infty}\frac{\sum_{l=1}^L\sum_{k=1}^{K-1}\sum_{i=k+1}^{K}R_{ki}^l}{\log
  SNR}\le N. \label{eq:2}
\end{equation}

We now have an upper bound for half of the messages from all
users. We can bound the DoF for the rest
of the messages by enhancing the received signal of user $K$ and
provide genie information to user $K$ in a similar fashion, as
illustrated in Fig. \ref{fig-outerbnd_2}, which yields
\begin{equation}
\lim_{SNR\to
  \infty}\frac{\sum_{l=1}^L\sum_{k=2}^{K}\sum_{i=1}^{k-1}R_{ki}^l}{\log
  SNR}\le N.\label{eq:3}
\end{equation}

Given (\ref{eq:2}) and (\ref{eq:3}), we have
\begin{equation}
DoF\le 2N.
\end{equation}
\end{proof}

\subsection{Optimality of Achievable DoF for Special Cases in Previous Work}\label{sec:optim-achi-dof}
We can now evaluate the optimality of the achievable DoF for some special cases of the
MIMO multi-way relay channel provided in previous work using the newly derived DoF upperbound.

\begin{itemize}
\item {\it MIMO $K$-user $Y$ channel \cite{KYchannel}}: This is the MIMO multi-way relay channel
  with one cluster that contains $K$ users. User $i$ has $M_i$
  antennas and the relay has $N$ antennas. It was shown in \cite{KYchannel} that each user can send $K-1$
  independent data streams with DoF $d$ for each stream if
  \begin{equation}
    M_i\ge d(K-1),~N\ge \frac{dK(K-1)}{2}\nonumber
  \end{equation}
  \begin{equation}
    N<\min\{M_i+M_j-d | \forall i\neq j\}.
  \end{equation}
For this case, our DoF upperbound specializes to 
\begin{equation}
DoF\le \min\{\sum_{i=1}^KM_i,2N\}.
\end{equation} If we have $M_i\ge d(K-1)$ and fix $N= \frac{dK(K-1)}{2}$,
the DoF upper bound becomes
\begin{equation}
DoF\le dK(K-1).
\end{equation} If we further have $M_i>\frac{K^2-K+2}{4}$, the condition
$N<\min\{M_i+M_j-d|\forall i\neq j\}$ is also satisfied, and the DoF upper
bound implies that the achievable DoF in \cite{KYchannel}, $dK(K-1)$, is indeed the optimal DoF.
\item {\it MIMO $K$-pair two-way relay channel
  \cite{IAinMUTWRC}}: This corresponds to the MIMO multi-way relay channel with $K$ clusters
  each with two users. Each user has $M$ antennas and wants to
  transmit $d$ data streams with DoF $1$. The relay has $Kd$ antennas. To
  guarantee interference-free transmission, we need
  \begin{equation}
    K\le \frac{2M}{d}-1,
  \end{equation}
  and the achievable DoF is $2Kd$.

For
  this case, the DoF upper bound becomes
\begin{equation}
DoF \le \min\{2KM,2Kd \}.
\end{equation} When $K\le \frac{2M}{d}-1$, we have $2Kd\le 2KM$ for
$K\ge 2$. The achievable DoF $2Kd$ given in reference \cite{IAinMUTWRC} is indeed the optimal DoF.
\end{itemize}

We now have seen that the newly derived upper bound is useful for proving tight results for some
special cases of the MIMO multi-way relay channel. In the next sections, we
utilize the upper bound to investigate the DoF of the more general MIMO multi-way relay channel, and provide our DoF findings.

\section{Two-Cluster MIMO multi-way Relay Channel}\label{sec:two-cluster-multiway}
With the newly derived DoF upper bound at hand, we now investigate the
achievable DoF for the general MIMO multi-way relay channel. We first focus
on the two-cluster MIMO multi-way relay channel. For the two-cluster case,
the only known result is the constant gap capacity result for the SISO
case \cite{divide} and the DoF for the two-user symmetric case
\cite{IAinMUTWRC}, i.e., users have the same number of antennas. Both
results are obtained using signal space alignment, or using techniques
that are in essence similar to signal space alignment such as using nested
lattice codes \cite{Uri_AWGN}. When the users have arbitrary number of antennas, the
optimal DoF has been unknown to date.

We first present the following lemma which characterizes the dimension of shared
signal space at the relay between two users with arbitrary number of antennas.

\begin{lemma} \label{lma}
For matrices ${\bf H}_1\in \mathbb{C}^{p\times q_1}$ and ${\bf
  H}_2\in \mathbb{C}^{p\times q_2}$, which have full rank almost
surely, the shared dimension of their column space can be specified as
follows. Note that without loss of generality we
assume $q_1\ge q_2$.

{\underline{Condition 1}}: If $p\ge q_1 \ge q_2$ and $q_1+q_2>p$, then there exist
$q_1+q_2-p$ non-zero linearly independent vectors ${\bf v}_i$ almost surely
such that we can find another two sets of linearly independent vectors ${\bf u}_i$ and
${\bf w}_i$, $i=1,\cdots,q_1+q_2-p$ such that
\begin{equation}
{\bf v}_i={\bf H}_1{\bf u}_{i}={\bf H}_2{\bf w}_{i}. \label{eq:12}
\end{equation}

{\underline{Condition 2}}: If $q_1\ge p \ge q_2$, then there exist
$q_2$ linearly independent vectors ${\bf v}_i$ almost surely
such that we can find another two sets of linearly independent vectors ${\bf u}_i$ and
${\bf w}_i$, $i=1,\cdots,q_2$ such that
\begin{equation}
{\bf v}_i={\bf H}_1{\bf u}_{i}={\bf H}_2{\bf w}_{i}\label{eq:16}
\end{equation}
\end{lemma}
\begin{proof}
The proof is provided in Appendix \ref{sec:proof-lemma-reflma}.
\end{proof}
\begin{remark}
The result in Condition 1 is the same as the result in
\cite{signalspaceY}. For this case, the dimension of the shared signal
space decreases as the number of antennas at the relay increases. For
this reason, as we will see in the achievable DoF for the MIMO multi-way
relay channel, using only a subset of the antennas at the relay is sufficient.
\end{remark}
\begin{remark}
Condition 2 implies that, when there is asymmetry
between the number of antennas at the users, the dimension of shared signal space
between two users cannot exceed the dimension of the user with the
smallest number of antennas.
\end{remark}

\subsection{Two Users in Each Cluster: Transmission Strategies}
For the two-cluster multi-way relay channel with two users in each
cluster, we can show that the following DoF is achievable:
\begin{proposition}\label{Prop-2cluster-MW-Ach-DoF}

({\bf i}) {\bf When $N\le M_2^1+M_2^2$, the following DoF is achievable}:
\begin{itemize}
\item {\it Case 1:} $N\le M_2^1$, $DoF=2N$.
\item {\it Case 2:} $N>M_2^1\ge M_2^2$ \begin{itemize}
                     \item {\it Condition 1:} $N\le M_1^1$ and $N\le M_1^2$, $DoF = 2N$.
                     \item {\it Condition 2:} $M_1^1<N\le M_1^2$ \begin{enumerate}
                                          \item $M_1^1+M_2^1+M_2^2\ge 2N$,
                                            $DoF=2N$.
                                          \item $M_1^1+M_2^1+M_2^2<2N$,
                                            $DoF=M_1^1+M_2^1+M_2^2$.
                                          \end{enumerate}
                     \item {\it Condition 3:} $M_1^2<N\le M_1^1$ \begin{enumerate}
                                          \item $M_2^1+M_1^2+M_2^2\ge 2N$,
                                            $DoF=2N$.
                                          \item $M_2^1+M_1^2+M_2^2<2N$,
                                            $DoF=\max\{M_2^1+M_1^2+M_2^2,N+M_2^1\}$.
                                          \end{enumerate}
                     \item {\it Condition 4:} $N>M_1^1,N> M_1^2$ \begin{enumerate}
                                          \item $M_1^1+M_2^1+M_1^2+M_2^2\ge 3N$,
                                            $DoF=2N$.
                                          \item $M_1^1+M_2^1+M_1^2+M_2^2<3N$,
                                            \begin{align}
                                              DoF=\max&\Biggl\{(M_1^2+M_2^2-N)^++N,
                                              M_1^1+M_2^1,\min\Bigl\{\frac{2(M_1^1+M_2^1+M_1^2+M_2^2)}{3},\nonumber\\
                                              &\frac{4(M_1^1+M_2^1+M_2^2)-2M_1^2}{3},\frac{4(M_2^1+M_1^2+M_2^2)-2M_1^1}{3}\Bigr\}\Biggr\}.\label{eq:9}
                                            \end{align}
                                          \end{enumerate}

                     \end{itemize}
\end{itemize}

({\bf ii}) {\bf When $N> M_2^1+M_2^2$, the following DoF is achievable}:
\begin{itemize}
\item {\it Case 1:} $N\ge 2(M_2^1+M_2^2)$, $DoF=2(M_2^1+M_2^2)$.
\item {\it Case 2:} $N<2(M_2^1+M_2^2)\le 2(M_2^1+M_2^1) \le 2(M_1^1+M_2^1)$
                     \begin{itemize}
                     \item {\it Condition 1:} $N\le M_1^1$ and $N\le M_1^2$, $DoF = 2(M_2^1+M_2^2)$.
                     \item {\it Condition 2:} $M_1^1<N\le M_1^2$, which implies
                       $M_1^1+M_2^1+M_2^2<2N$,
                                          \begin{enumerate}
                                          \item $N\ge 2M_2^1+M_2^2$,
                                            $DoF=2(M_2^1+M_2^2)$.
                                          \item $M_1^1\ge
                                            M_2^1+M_2^2$,
                                            $DoF=2(M_2^1+M_2^2)$.
                                          \item $N< 2M_2^1+M_2^2$ and $M_1^1<
                                            M_2^1+M_2^2$, $DoF=\max
                                           \{N+M_2^2, M_1^1+M_2^1+M_2^2\}$.
                                          \end{enumerate}
                     \item {\it Condition 3:} $M_1^2<N\le M_1^1$, which implies
                       $M_2^1+M_1^2+M_2^2<2N$,
                                          \begin{enumerate}
                                          \item $N\ge M_2^1+2M_2^2$,
                                            $DoF=2(M_2^1+M_2^2)$.
                                          \item $M_1^2\ge
                                            M_2^1+M_2^2$,
                                            $DoF=2(M_2^1+M_2^2)$.
                                          \item $N< M_2^1+2M_2^2$ and $M_1^2<M_2^1+M_2^2$,
                                            $DoF=\max\{M_2^1+M_1^2+M_2^2,N+M_2^1\}$.
                                          \end{enumerate}
                     \item {\it Condition 4:} $N>M_1^1,N> M_1^2$, which implies
                       $M_1^1+M_2^1+M_1^2+M_2^2<3N$,
                                          \begin{enumerate}
                                            \item $M_1^1\ge
                                              M_2^1+M_2^2$ and
                                              $M_1^2\ge M_2^1+M_2^2$, $DoF=2(M_2^1+M_2^2)$.
                                            \item $M_1^2\ge
                                              2M_2^1+M_2^2$, $DoF=2(M_2^1+M_2^2)$.
                                            \item $M_1^1\ge
                                              M_2^1+2M_2^2$,
                                              $DoF=2(M_2^1+M_2^2)$.
                                            \item Otherwise,
                                              \begin{align}
                                              DoF=\max&\Biggl\{N,(M_1^1+M_2^1-N)^++N,(M_1^2+M_2^2-N)^++N,\nonumber\\
                                              &\min\Bigl\{\frac{2(M_1^1+M_2^1+M_1^2+M_2^2)}{3},
                                              \frac{4(M_1^1+M_2^1+M_2^2)-2M_1^2}{3},\nonumber\\
                                              &\frac{4(M_2^1+M_1^2+M_2^2)-2M_1^1}{3}\Bigr\}\Biggr\}.\label{eq:10}
                                            \end{align}
                                          \end{enumerate}

                     \end{itemize}
\end{itemize}
\end{proposition}

\begin{proof}
We next provide detailed transmission schemes to show how the above
DoF can be achieved and identify scenarios when the DoF upper bound
can be achieved.

{\bf (i)} {\bf When $N\le M_2^1+M_2^2$}:

Under this condition, the DoF upper bound in equation \eqref{eq:4} reduces to
\begin{equation}
DoF\le 2N.
\end{equation} We further consider the following cases:

{\it Case 1: $N\le M_2^1$:} This condition corresponds to the case when the
  relay always has less antennas than both users in at least one of
  the two clusters. The DoF $2N$ can be achieved by only allowing the
  users in the cluster with more antennas than the relay to exchange information,
  which yields a two-way relay channel. Since both users have more
  antennas than the relay, they can perfectly align $N$ independent
  data streams at the relay. The
  functional-decode-and-forward (FDF) strategy can thus achieve the
  DoF upper bound $2N$.

{\it Case 2: $N>M_2^1\ge M_2^2$:} This condition corresponds to the
case when the relay has more antennas than at least one users in each
  cluster. A single pair of users thus cannot perfectly align $N$ independent data
  streams at the relay. However, it is still possible to achieve the
  optimal DoF by allowing two clusters of users to use the relay at
  the same time. Depending on the number of antennas
  at the relay and the users, we further consider the
  following conditions:

  {\it Condition 1: $N\le M_1^1$ and $N\le M_1^2$.} For this case, one of the
    users in each cluster has more antennas than the relay. From {Condition 2} in {\it
    Lemma \ref{lma}}, if we set ${\bf H}_{1}={\bf H}_{R(1,1)}$ and ${\bf
    H}_{2}={\bf H}_{R(2,1)}$, we can see that for user 1 and user 2 in
    cluster 1, they can find $M_2^1$ non-zero linearly independent
    vectors ${\bf q}_{1i}, {\bf v}_{(1,1)i}$ and ${\bf v}_{(2,1)i}$, $i=1,\cdots,M_2^1$ for
    cluster 1 such that
    \begin{equation}
      {\bf H}_{R(1,1)}{\bf v}_{(1,1)i}={\bf H}_{R(2,1)}{\bf v}_{(2,1)i}={\bf q}_{1i} \label{eq:5}.
    \end{equation}
    This means that user 1 and user 2 can share $M_2^1$ dimensional
    space at the relay. Following the same argument, we can see that
    user 1 and user 2 in cluster 2 share $M_2^2$ dimensional
    space at the relay, i.e., they can find $M_2^2$ non-zero linearly
    independent vectors ${\bf q}_{2i}, {\bf v}_{(1,2)i}$ and ${\bf
      v}_{(2,2)i}$ such that
    \begin{equation}
      {\bf H}_{R(1,2)}{\bf v}_{(1,2)i}={\bf H}_{R(2,2)}{\bf
        v}_{(2,2)i}={\bf q}_{2i} \label{eq:6}.
    \end{equation}

    Since we have $M_2^1+M_2^2\ge N$, the users in cluster 1 can
    choose $M_2^{1\prime}$ vectors out of the vectors ${\bf q}_{1i}$,
    and the users in cluster 2 can choose $M_2^{2\prime}$ vectors out of the vectors ${\bf q}_{2i}$, such that
    these vectors are linearly independent almost surely and
    $M_2^{1\prime}+M_2^{2\prime}=N$, as their target signal directions
    at the relay. We denote the set of vectors chosen by cluster 1 as
    $\mathcal{Q}_1$ and the set of vectors chosen by cluster 2 as $\mathcal{Q}_2$.

    Based on the above analysis, we can construct the transmission scheme as
    follows: User $1$ and user $2$ in cluster 1 send $M_2^{1\prime}$ independent data
    streams $d_{1i}^1$ and $d_{2i}^1$ along the directions ${\bf v}_{(1,1)i}$ and ${\bf v}_{(2,1)i}$,
    respectively. User $1$ and user $2$ in cluster 2 send $M_2^{2\prime}$ independent data
    streams $d_{1i}^2$ and $d_{2i}^2$ along the directions ${\bf v}_{(1,2)i}$ and ${\bf v}_{(2,2)i}$,
    respectively. We have
    \begin{equation}
      {\bf X}_{k,1}=\sum_{i\in \mathcal{Q}_1}{\bf v}_{(k,1)i}d_{ki}^1, k=1,2,
    \end{equation}
    \begin{equation}
      {\bf X}_{k,2}=\sum_{i\in \mathcal{Q}_2}{\bf v}_{(k,2)i}d_{ki}^2, k=1,2.
    \end{equation}

    The received signal at the relay is
    \begin{align}
      {\bf Y}_R&=\sum_{k=1}^2{\bf H}_{R(k,1)}{\bf X}_{k,1}+\sum_{k=1}^2{\bf H}_{R(k,2)}{\bf X}_{k,2}\\
      &=\sum_{i\in \mathcal{Q}_1}{\bf
        q}_{1i}(d_{1i}^1+d_{2i}^1)+\sum_{i\in \mathcal{Q}_2}{\bf q}_{2i}(d_{1i}^2+d_{2i}^2)
    \end{align}

    The relay can then decode $d_{1i}^1+d_{2i}^1$ and
    $d_{1i}^2+d_{2i}^2$ using zero forcing.

    The relay now needs to transmit $d_{1i}^1+d_{2i}^1$ to user 1 and
    user 2 in cluster 1
    and also transmit $d_{1i}^2+d_{2i}^2$ to user 1 and
    user 2 in cluster 2. For this end, we
    let the users apply a receiver-side filter ${\bf u}_{(k,l)i}$
    such that
    \begin{equation}
      ({\bf u}_{(1,1)i})^T{\bf H}_{(1,1)R}=({\bf u}_{(2,1)i})^T{\bf H}_{(2,1)R}={\bf g}_{1i}^T,\label{eq:7}
    \end{equation}
    \begin{equation}
      ({\bf u}_{(1,2)i})^T{\bf H}_{(1,2)R}=({\bf u}_{(2,2)i})^T{\bf H}_{(2,2)R}={\bf g}_{2i}^T,\label{eq:8}
    \end{equation} which makes the users in one cluster appear to be the
    same user to the relay.

    Taking transpose of equations (\ref{eq:7}) and (\ref{eq:8}), we
    can see that the problem of finding the vectors ${\bf u}_{(k,l)i}$
    are the same problem as finding the vectors ${\bf
      v}_{(k,l)i}$. Therefore the users in cluster $l$ can find $M_2^l$
    such triplets of non-zero linearly independent vectors
    \begin{equation}
      ({\bf u}_{(1,l)i},{\bf u}_{(2,l)i},{\bf g}_{li}).
    \end{equation} The users in cluster $1$ can then choose
    $M_2^{1\prime}$ vectors ${\bf g}_{1i}$ and the users in cluster
    $2$ can choose $M_2^{2\prime}$ vectors ${\bf g}_{2i}$, such that
    they are all linearly independent, as their target directions to
    receive signals transmitted from the relay. Using these chosen
    vectors, user $k$ in cluster $l$ can form a beamforming matrix
    ${\bf U}_{k,l}$, which has the chosen ${\bf u}_{(k,l)i}$ vectors
    as its rows, and apply it to the received signals:
    \begin{align}
      {\bf Y}_{k,l}^\prime
      &={\bf U}_{k,l}{\bf Y}_{k,l}\\
      &={\bf U}_{k,l}{\bf H}_{(k,l)R}{\bf X}_R+{\bf U}_{k,l}{\bf
        Z}_{k,l}\\
      &={\bf G}_l{\bf X}_R+{\bf U}_{k,l}{\bf
        Z}_{k,l}
    \end{align} where the matrix ${\bf G}_l$ is of dimension
    $M_2^{l\prime}\times N$ and has the chosen vectors ${\bf g}_{li}$
    as its rows. The relay can use zero-forcing precoding to
    communicate  $d_{1i}^1+d_{2i}^1$ and $d_{1i}^2+d_{2i}^2$ to the
    intended users. The users can now subtract their own side
    information from the received signals to decode the intended
    messages. Therefore the DoF $2N$ is achievable.

  {\it Condition 2: $M_1^1<N\le M_1^2$.} For this case, cluster 2 has a
    user with more antennas than the relay while both users in cluster
    1 have less antennas than the relay. From Condition 1
    in {\it Lemma \ref{lma}}, the users in cluster 1 can share
    $M_1^1+M_2^1-N$ dimensional signal space at the relay, and from
    Condition 2 in {\it Lemma \ref{lma}}, the
    users in cluster 2 can share $M_2^2$ dimensional signal space at
    the relay. Note that since $N\le M_2^1+M_2^2$ and we assume
    $M_2^1\ge M_2^2$, we have $N\le
    M_1^1+M_2^1$, i.e., $M_1^1+M_2^1-N$ is always greater than
    zero. This leads to the following two cases that we need to investigate:

        $1)$ $M_1^1+M_2^1+M_2^2\ge 2N$. For this case, the total
        dimension of the shared signal space for the two clusters
        exceeds the available dimension available at the relay. The
        transmission scheme for the case $N\le M_1^1$ and $N\le M_1^2$
        can be used to achieve the DoF $2N$. Note that for this case,
        user 1 in cluster 2 has more antennas than the relay, and
        therefore it can send signals targeted at any signal
        dimension at the relay. User 2 in cluster 2 can transmit its
        data streams using some random beamforming vectors, and user 1
        in cluster 2 can control the direction of its transmitted
        data streams such that they arrive aligned with the data
        streams sent by user 2 in cluster 2. Users in cluster 1, on
        the other hand, need to design their beamforming vectors
        jointly such that their data streams are aligned at the
        relay. The received data streams from cluster 1 and cluster
        2 are linearly independent at the relay almost surely since
        the channel matrices are generated from a continuous
        distribution. The relay then decodes the sum of the messages
        from each clusters, and broadcasts the messages back to the
        intended clusters with proper receiver-side filtering at the
        users. The detailed scheme is similar to the previous case and
        is thus omitted.

        $2)$ $M_1^1+M_2^1+M_2^2< 2N$. Under this condition, the signal space
        available at the relay cannot be fully utilized by the two
        clusters, because the total dimension of shared signal space
        for the two clusters is $M_1^1+M_2^1-N+M_2^2$, which is
        smaller than $N$. Therefore the DoF upper bound $2N$ cannot be
        achieved using signal space alignment.

        For this case, if $\frac{M_1^1+M_2^1+M_2^2}{2}$ is
        an integer, we can let the relay to use
        $N^\prime=\frac{M_1^1+M_2^1+M_2^2}{2}$ antennas to assist the
        users. It is easy to see that $N^\prime\ge M_2^1$,
        $M_1^2\ge N^\prime\ge M_2^2$, and $N^\prime\ge M_1^1$ since $M_1^1<N\le
        M_2^1+M_2^2$. By using only a subset of the antennas at the
        relay, users in cluster 1 can still share $M_1^1+M_2^1-N$
        dimensional space and users in cluster 2 can still share $M_2^2$
        dimensional space. Since we also have $M_1^1+M_2^1-N+M_2^2=
        N^\prime$, using the schemes described in
        the previous part, we can achieve the DoF
        $M_1^1+M_2^1+M_2^2$.

        If $\frac{M_1^1+M_2^1+M_2^2}{2}$ is not an integer, we can use
        a two-symbol extension to create an effectively two-cluster
        MIMO multi-way relay channel with $2M_1^1,2M_2^1,2M_1^2,2M_2^2,2N$
        antennas at the users and the relay, respectively, and using
        the same argument as in the case when
        $\frac{M_1^1+M_2^1+M_2^2}{2}$ is an integer, we can achieve
        the DoF $M_1^1+M_2^1+M_2^2$ per channel use.
\begin{remark}\label{remarkon_ssa_relayantennas}
        Note that under this condition $M_1^1+M_2^1+M_2^2<2N$, an alternative scheme is to
        let the relay use $N$ antennas to assist the users. The
        users in cluster 2 can still share the $M_2^2$ dimensional
        signal space at the relay. The users in cluster 1 can
        use the shared $M_1^1+M_2^1-N$ dimensional space for signal
        space alignment, which
        yields an achievable DoF $2(M_1^1+M_2^1+M_2^2)-2N$ or use
        the rest $N-M_2^2\le M_2^1$ dimensional space in the multiple-access
        fashion, which yields an achievable DoF $N+M_2^2$. It is
        easy to see that $M_1^1+M_2^1+M_2^2$, which is achieved by
        using a subset of antennas at the relay, is the largest
        achievable DoF. This is because using more antennas at the
        relay decreases the number of dimension that can be shared by
        users using signal space alignment. The additional signal
        space, on the other hand, can only be used by
        a single user if signal space alignment is not used. Adding one
        antenna at the relay sacrifices two signal bits but only
        obtains one signal bit in return.
\end{remark}

{\it Condition 3: $M_1^2<N\le M_1^1$.} Based on {\it Lemma \ref{lma}}, the users in
cluster 1 share a $M_2^1$ dimensional signal space at the
relay and the users in cluster 2 share a
$M_1^2+M_2^2-N$ dimensional signal space at the relay. Different from
the case when $M_1^1<N\le M_1^2$, for users in
cluster 2, we cannot guarantee that $M_1^2+M_2^2-N$
is always positive. We further investigate the following cases:

$1)$ $M_2^1+M_1^2+M_2^2\ge 2N$. For this case, the total dimension of the shared
signal space of the two clusters exceeds the available dimension of
the signal space at the relay. The DoF $2N$ can thus be achieved using
signal space alignment, as described in the scheme for {\it Case 2 -
  Condition 2.(1), $N\le M_2^1+M_2^2$}.

$2)$ $M_2^1+M_1^2+M_2^2< 2N$. This condition implies that
\begin{equation}
M_1^2+M_2^2-N<\frac{M_1^2+M_2^2-M_2^1}{2}.
\end{equation}

When $M_1^2+M_2^2\le M_2^1$, we have $M_1^2+M_2^2-N<0$,
i.e., users in cluster 2 cannot share any signal space at the
relay. Therefore we let users in cluster 1 use the shared $M_2^1$
dimensional signal space to perform signal space alignment, and let
the users in cluster 2 use the rest $N-M_2^1$ dimensional signal space at the relay in the
multiple-access fashion. After decoding the sum of the messages from
cluster 1 and the individual messages from cluster 2, the relay can
then use zero-forcing precoding to broadcast the messages to the intended users with proper
receiver-side filtering at users in cluster 1. Using this scheme,
users in cluster 1 can exchange $2M_2^1$ messages and the users in
cluster 2 can exchange $N-M_2^1$ messages. We can achieve
DoF $N+M_2^1$.

When $M_1^2+M_2^2>M_2^1$, $M_1^2+M_2^2-N$ can be positive. For this
case, we can let the relay use $N^\prime=\frac{M_2^1+M_1^2+M_2^2}{2}$ antennas
to assist the users. Since we have $M_2^1+M_2^2\ge N > M_1^2$,
$N^\prime\ge M_1^2\ge M_2^2$. We also have $M_1^1>N^\prime>M_2^1$. Following the results in {\it Case 2 -
  Condition 2.(2), $N\le M_2^1+M_2^2$}, we can achieve the DoF
$M_2^1+M_1^2+M_2^2$.

We can also let the relay use all the antennas
to assist the users. If we allow the users in both clusters to use
signal space alignment, the achievable DoF is
$2(M_2^1+M_1^2+M_2^2)-2N$. It is easy to see that this achievable DoF is always
smaller than $M_2^1+M_1^2+M_2^2$ under the condition
$M_2^1+M_1^2+M_2^2<2N$. We can also let the users in cluster 1
use signal space alignment, but the users in cluster 2 use the relay in the
multiple-access fashion. This yields the achievable DoF $N+M_2^1$.

{\it Condition 4: $N>M_1^1$ and $N>M_1^2$.} Based on {\it Lemma
  \ref{lma}}, users in cluster 1 share a
$M_1^1+M_2^1-N$ dimensional signal space, and users in cluster 2 share
a $M_1^2+M_2^2-N$ dimensional signal space at the relay. Note that we always have
$M_1^1+M_2^1-N>0$ for $N\le M_2^1+M_2^2$. We further investigate the
following cases:

$1)$ $M_1^1+M_2^1+M_1^2+M_2^2\ge 3N$. For this case, the total dimension of
shared signal space for the two clusters exceeds the available signal
space at the relay. Both clusters can use signal space alignment to
achieve the DoF upper bound $2N$. The scheme can be designed in the
same fashion as in previous cases and the details are thus omitted.

$2)$ $M_1^1+M_2^1+M_1^2+M_2^2< 3N$. This condition implies that
\begin{equation}
M_1^2+M_2^2-N<\frac{2(M_1^2+M_2^2)-(M_1^1+M_2^1)}{3}.
\end{equation}

When $2(M_1^2+M_2^2)\le M_1^1+M_2^1$, $M_1^2+M_2^2-N$ is always less
than zero, i.e., there is no shared signal space at the relay for the
users in cluster 2. For this case, we let the relay use all the
antennas to assist the users. Users in cluster 1 can always share
the $M_1^1+M_2^1-N$ dimensional signal space at the relay. The
users in cluster 2 use the relay in the multiple-access fashion. This
yields the achievable DoF
\begin{equation}
2(M_1^1+M_2^1-N)+N-(M_1^1+M_2^1-N)=M_1^1+M_2^1.
\end{equation}

When $2(M_1^2+M_2^2)> M_1^1+M_2^1$, $M_1^2+M_2^2-N$ can be positive. For this case, we let
the relay use only $N^\prime=\frac{M_1^1+M_2^1+M_1^2+M_2^2}{3}$ antennas to
assist the users, if $\frac{M_1^1+M_2^1+M_1^2+M_2^2}{3}$ is an
integer. The case when $\frac{M_1^1+M_2^1+M_1^2+M_2^2}{3}$ is not an
integer can be addressed using symbol extension. It is easy to see that $N^\prime>M_2^1$ and
$N^\prime>M_2^2$. However, the relation between
$\frac{M_1^1+M_2^1+M_1^2+M_2^2}{3}$ and $M_1^1$ depends on
the relation between $M_2^1+M_1^2+M_2^2$ and $2M_1^1$; the relation between
$\frac{M_1^1+M_2^1+M_1^2+M_2^2}{3}$ and $M_1^2$ depends on
the relation between $M_1^1+M_2^1+M_2^2$ and $2M_1^2$:
\begin{itemize}
\item $M_2^1+M_1^2+M_2^2\ge 2M_1^1$ and $M_1^1+M_2^1+M_2^2\ge 2M_1^2$:
  For this case, users in cluster 1 share
  $\frac{2(M_1^1+M_2^1)-(M_1^2+M_2^2)}{3}$ dimensional signal space
  and users in cluster 2 share
  $\frac{2(M_1^2+M_2^2)-(M_1^1+M_2^1)}{3}$ dimensional signal
  space. The achievable DoF is $\frac{2(M_1^1+M_2^1+M_1^2+M_2^2)}{3}$.
\item $M_2^1+M_1^2+M_2^2\ge 2M_1^1$ and $M_1^1+M_2^1+M_2^2<2M_1^2$:
  For this case, users in cluster 1 share
  $\frac{2(M_1^1+M_2^1)-(M_1^2+M_2^2)}{3}$ dimensional signal space
  and users in cluster 2 share
  $M_2^2$ dimensional signal
  space. The achievable DoF is
  $\frac{4(M_1^1+M_2^1+M_2^2)-2M_1^2}{3}$.
\item $M_2^1+M_1^2+M_2^2< 2M_1^1$ and $M_1^1+M_2^1+M_2^2\ge 2M_1^2$:
  For this case, users in cluster 1 share
  $M_2^1$ dimensional signal space
  and users in cluster 2 share
  $\frac{2(M_1^2+M_2^2)-(M_1^1+M_2^1)}{3}$ dimensional signal
  space. The achievable DoF is
  $\frac{4(M_2^1+M_1^2+M_2^2)-2M_1^1}{3}$.
\item $M_2^1+M_1^2+M_2^2< 2M_1^1$ and $M_1^1+M_2^1+M_2^2<2M_1^2$: This
  case is not possible since the first condition implies $M_1^1>M_1^2$
  and the second condition implies $M_1^1<M_1^2$.
\end{itemize}

From the above cases, we can see that the achievable DoF is
\begin{equation}
\min\left\{\frac{2(M_1^1+M_2^1+M_1^2+M_2^2)}{3},\frac{4(M_1^1+M_2^1+M_2^2)-2M_1^2}{3},\frac{4(M_2^1+M_1^2+M_2^2)-2M_1^1}{3}\right\}
\end{equation}

Note that we can also let relay use all the antennas to assist the
users. We only allow cluster 1 to use signal space alignment, and let users in
cluster 2 use the relay in the multiple-access fashion. This yields
the achievable DoF $M_1^1+M_2^1$. In addition, we can also let
cluster 1 use the relay in the multiple-access fashion, and cluster 2
use signal space alignment. This yields the achievable DoF
$(M_1^2+M_2^2-N)^++N$. Combining both achievable DoF, we
have the desired result in equation \eqref{eq:9}.
\begin{remark}
Note that we can also use multiple-access transmission for both
clusters. However, the achievable DoF $N$ is always less than
$M_1^1+M_2^1$.
\end{remark}
\begin{remark}
If we let the relay use all the
antennas and use signal space alignment, the achievable DoF
$2(M_1^1+M_2^1-N)+2(M_1^2+M_2^2-N)^+$ is also smaller than the
achievable DoF in \eqref{eq:9}. Similar to {\it Remark
  \ref{remarkon_ssa_relayantennas}}, this is because when using signal
space alignment, increase the number of antennas at the relay will
decrease the dimension of shared signal space for the users. Using too
many antennas at the relay will reduce the dimension of shared signal
space and result in some unused signal space,
when only signal space alignment is used. It is always more desirable to
use the exact number of antennas at the relay such that all spatial
dimension is occupied for signal space alignment.
\end{remark}

({\bf ii}) {\bf When $N > M_2^1+M_2^2$}: 

Under this setting, the DoF upper bound in equation \eqref{eq:4} reduces to
\begin{equation}
DoF\le 2(M_2^1+M_2^2).
\end{equation}

{\it Case 1: $N\ge 2(M_2^1+M_2^2)$}

The DoF upper bound can be easily achieved for this case since the
relay has enough antennas to perform zero-forcing decoding and
precoding. Since we have $M_1^1\ge M_2^1$, $M_1^2\ge M_2^2$, we can
let user 1 use only $M_2^1$ of its antennas and let user 3 use only
$M_2^2$ of its antennas to transmit. The relay can decode all the messages and
broadcast the messages to the intended users since it has sufficient
spatial dimension.

{\it Case 2: $N<2(M_2^1+M_2^2)$}

For this case, we also have that $N\le 2(M_2^1+M_2^1)\le
2(M_1^1+M_2^1)$. Depending on the number of antennas at the
users and the relay, we need to further consider the following
conditions:

{\it Condition 1:} $N\le M_1^1$ and $N\le M_1^2$. From Condition 2 in
{\it Lemma \ref{lma}}, the users in cluster 1 share $M_2^1$
dimensional signal space and the users in cluster 2 share $M_2^2$
dimensional space at the relay. Since we also have
$N>M_2^1+M_2^2$, the users in each cluster can fully utilize their
shared signal space at the relay to exchange messages. Specifically,
users in cluster 1 and cluster 2 can transmit $M_2^1$ and $M_2^2$ data streams such that they are
aligned at the relay, respectively. The relay decodes the sum of the messages and
broadcast back to the intended clusters with proper receiver-side
processing at the users. The DoF upper bound $2(M_2^1+M_2^2)$ can be
achieved. The detailed scheme is similar to the previous cases and
thus is omitted.

{\it Condition 2:} $M_1^1<N\le M_1^2$. From {\it Lemma \ref{lma}},
users in cluster 1 share $(M_1^1+M_2^1-N)^+$ dimensional signal space,
while the users in cluster 2 share $M_2^2$ dimensional signal
space. Under this condition, we have $M_1^1+M_2^1-N<M_2^1$. Therefore
using all the antennas at the relay and signal space alignment at two
clusters cannot achieve the optimal DoF. In addition, since we have $N > M_2^1+M_2^2$, $M_1^1<N$ implies that
$M_1^1+M_2^1+M_2^2<2N$, and thus the total dimension of shared signal
space for the two clusters is less than $N$. Note that we do not have
$M_1^1+M_2^1-N\ge 0$ for $N>M_2^1+M_2^2$, which is different from the
case when $N\le M_2^1+M_2^2$.

$1)$ We first consider a scheme that
allows the users in cluster 1 use the relay in the multiple-access fashion, and we also let the users in cluster 2 to use the shared $M_2^2$
dimensional space to perform signal space alignment. The
dimension of signal space available for cluster 1 is $N-M_2^2$. As
long as $N-M_2^2\ge 2M_2^1$, the users in cluster 1 can still exchange
a total of $2M_2^1$ messages using multiple-access type of
schemes. The users in cluster 2 can always exchange $2M_2^2$ messages
using signal space alignment. Therefore, when $N\ge 2M_2^1+M_2^2$, we can still achieve the
DoF upper bound $2(M_2^1+M_2^2)$.

$2)$ We next only allow the relay to use a subset of the antennas to
assist the users. Specifically, if we have $M_1^1\ge M_2^1+M_2^2$, we
can let the relay use exactly $M_2^1+M_2^2$ antennas. From the result
in {\it Case 2 - Condition 1, $N\le M_2^1+M_2^2$}, we can achieve
DoF $2(M_2^1+M_2^2)$, which matches the upper bound.

$3)$ When the conditions in the above cases are not satisfied, the scheme used achieves the DoF
\begin{equation}
DoF=\max\{N+M_2^2,M_1^1+M_2^1+M_2^2\}
\end{equation}
\begin{remark}
The achievable schemes for the optimal DoF under this condition
imply that using signal space alignment to let the users share the
signal space of the relay is not always the optimal
approach. Depending on the number of antennas at the users and the
relay, multiple-access transmission or a combination of both can be
more beneficial.
\end{remark}

{\it Condition 3:} $M_1^2<N\le M_1^1$. The result for this case can be
obtained following similar arguments from {\it Case 2 - Condition 2,
  $N>M_2^1+M_2^2$} and the details are omitted.

{\it Condition 4:} $N>M_1^1$ and $N>M_1^2$. Under this condition, we
consider the following cases:

  $1)$ $M_1^1\ge M_2^1+M_2^2$ and $M_1^2\ge M_2^1+M_2^2$: For this
  case, we have $N>M_2^1+M_2^2$. We can let the relay use only
  $M_2^1+M_2^2$ antennas to assist the users. This case is then
  reduced to {\it Case 2 - Condition 2.(2), $N\le M_2^1+M_2^2$} or {\it Case 2 - Condition
    3.(2), $N\le M_2^1+M_2^2$}. The optimal DoF $2(M_2^1+M_2^2)$ can thus be achieved.

  $2)$ $M_1^2\ge 2M_2^1+M_2^2$: For this case, we have
  $N>2M_2^1+M_2^2$. We can let the relay use only $2M_2^1+M_2^2$
  antennas to assist the users. Based on {\it Lemma
    \ref{lma}}, the users in cluster 2 share $M_2^2$ dimensional
  signal space at the relay. We let the users in cluster 2 to use
  signal space alignment to exchange $2M_2^2$ messages using $M_2^2$
  dimensional space, while the users
  in cluster 1 use the rest $2M_2^1$ dimensional space at the relay in the multiple-access fashion to
  exchange $2M_2^1$ messages. The achieved DoF is thus $2(M_2^1+M_2^2)$,
  which matches the upper bound.

  $3)$ $M_1^1\ge M_2^1+2M_2^2$: For this case, the optimal DoF
  $2(M_2^1+M_2^2)$ can be achieved following the same argument as in
  the previous case when $M_1^2\ge 2M_2^1+M_2^2$.

  $4)$ For the other cases, we can always achieve the DoF $N$ by
  letting all the users transmit the data streams to the relay in the
  multiple-access fashion, and relay decodes all data streams and
  broadcasts back to the users. We can also let one cluster of users
  use signal space alignment, and the other cluster of users use
  multiple-access transmission. For example, we can let cluster 1 use
  $(M_1^1+M_2^1-N)^+$ dimensional space to perform signal space
  alignment, and cluster 2 use the rest $N-(M_1^1+M_2^1-N)^+$
  dimensional space to perform multiple-access transmission. The DoF
  $(M_1^1+M_2^1-N)^++N$ can thus be achieved. The DoF
  $(M_1^1+M_2^1-N)^++N$ can be achieved in a similar fashion. The last
  term in equation \eqref{eq:10} can be achieved by
  using $\frac{M_1^1+M_2^1+M_1^2+M_2^2}{3}$ antennas at the relay, and
  the analysis is similar to {\it Case 2 - Condition 4.(2), $N\le
    M_2^1+M_2^2$} and the details are omitted.
\end{proof}

\subsection{Two Users in Each Cluster: Optimal DoF}
We now have presented a set of achievable DoF for the general
two-cluster MIMO multi-way relay channel with two users in each
cluster. With the achievable DoF in {\it Proposition
  \ref{Prop-2cluster-MW-Ach-DoF}} and the DoF upper bound in {\it
  Theorem \ref{thm-DoF-upperbound}}, we can establish the optimal DoF.

\begin{theorem}\label{thm-2cluster-2user}
Consider the two-cluster MIMO multi-way relay channel with two users in each
cluster as described in Section \ref{sec:system-model}. The optimal DoF is described as follows:

{\it (1) $DoF^*=2N$}
\begin{itemize}
\item When $N\le M_2^1$, or
\item When $M_2^1+M_2^2 \ge N>M_2^1\ge M_2^2$ and one of the following conditions is satisfied:
                     \begin{itemize}
                     \item  $N\le M_1^1$ and $N\le M_1^2$.
                     \item  $M_1^1<N\le M_1^2$ and
                                          $M_1^1+M_2^1+M_2^2\ge 2N$.

                     \item  $M_1^2<N\le M_1^1$ and $M_2^1+M_1^2+M_2^2\ge 2N$.

                     \item  $N>M_1^1,N> M_1^2$ and $M_1^1+M_2^1+M_1^2+M_2^2\ge 3N$.
                     \end{itemize}
\end{itemize}

{\it (2) $DoF^*=2(M_2^1+M_2^2)$}
\begin{itemize}
\item When $N\ge 2(M_2^1+M_2^2)$, or
\item When $M_2^1+M_2^2<N<2(M_2^1+M_2^2)$ and one of the following
  conditions is satisfied:
                     \begin{itemize}
                     \item $N\le M_1^1$, $N\le M_1^2$.
                     \item $M_1^1<N\le M_1^2$ and
                                          \begin{itemize}
                                          \item $N\ge 2M_2^1+M_2^2$ or
                                          \item $M_1^1\ge M_2^1+M_2^2$.
                                          \end{itemize}
                     \item $M_1^2<N\le M_1^1$ and
                                          \begin{itemize}
                                          \item $N\ge M_2^1+2M_2^2$ or
                                          \item $M_1^2\ge
                                            M_2^1+M_2^2$.
                                          \end{itemize}
                     \item $N>M_1^1,N> M_1^2$ and
                                          \begin{itemize}
                                            \item $M_1^1\ge
                                              M_2^1+M_2^2$,
                                              $M_1^2\ge M_2^1+M_2^2$ or
                                            \item $M_1^2\ge
                                              2M_2^1+M_2^2$ or
                                            \item $M_1^1\ge
                                              M_2^1+2M_2^2$.
                                          \end{itemize}

                     \end{itemize}
\end{itemize}
\end{theorem}

\subsection{Two Users in Each Cluster: Symmetric Case} \label{sec:two-cluster-multiway-sym}
We now consider the two-cluster MIMO multi-way relay channel with two users
in each cluster with $M_1^1=M_1^2=M_1$ and $M_2^1=M_2^2=M_2$. The
optimal DoF for this special case is summarized as follows:

\begin{corollary}
For the two-cluster MIMO multi-way relay channel with two users in each
cluster with $M_1^1=M_1^2=M_1$ and $M_2^1=M_2^2=M_2$ (without loss of
generality assume $M_1\ge M_2$), the optimal DoF
is:

When $N\le 2M_2$,
\begin{itemize}
\item $N\le M_2$, $DoF^*=2N$.
\item $M_2<N\le M_1$, $DoF^*=2N$.
\item $M_2\le M_1<N\le \frac{2}{3}(M_1+M_2)$, $DoF^*=2N$.
\end{itemize}

When $N>2M_2$,
\begin{itemize}
\item $N\ge 4M_2$, $DoF^*=4M_2$.
\item $N<4M_2, N\le M_1$, $DoF^*=4M_2$.
\item $N>M_1 \ge 2M_2$, $DoF^*=4M_2$.
\end{itemize}
\end{corollary}
\begin{proof}
This corollary follows as a special case from {\it Proposition
  \ref{Prop-2cluster-MW-Ach-DoF}}, and the upperbound in {\it Theorem \ref{thm-DoF-upperbound}}.
\end{proof}

\begin{figure}[t]
\centering
\includegraphics[width=5in]{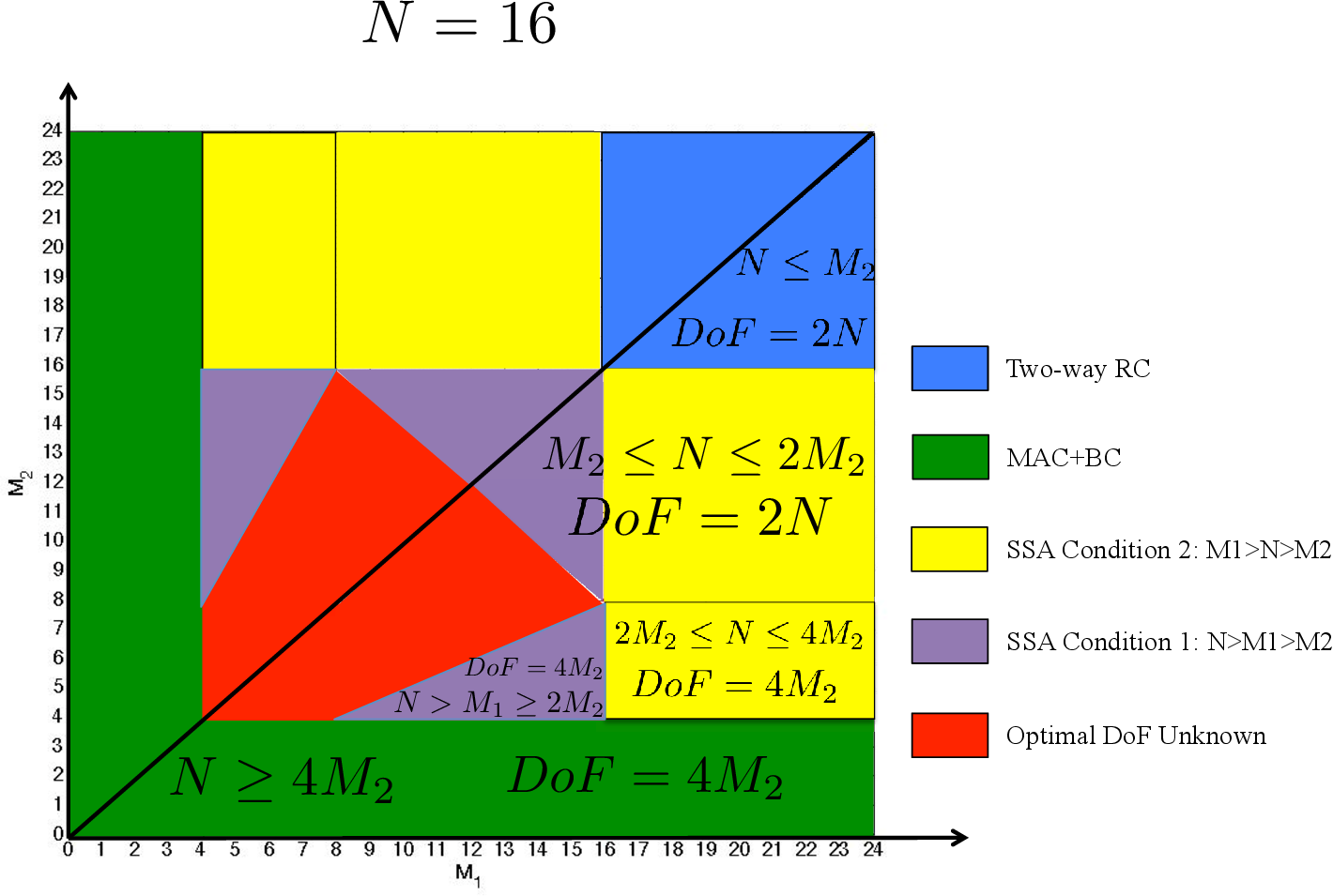}
\caption{Illustration for cases when DoF upper bound can be achieved.} \label{fig-illusSSA}
\end{figure}

Fig. \ref{fig-illusSSA} illustrates the regimes for which we can establish the optimal DoF for the
symmetric case with $N=16$. In the figure, {\it Two-way RC} denotes the region where the DoF can be achieved by only allowing one cluster to exchange data with the relay, which reduces the channel
to a two-way relay channel. {\it MAC+BC} denotes the region that
the users use multiple-access transmission and the relay
decodes and broadcasts the messages to the intended users. {\it
  SSA} represents signal space alignment, where different SSA conditions depend on the number of antennas at the users and the relay. Note that the different SSA conditions correspond to those introduced in {\it Lemma \ref{lma}}.

\subsection{Three Users in Each Cluster: General Case}
We now study the case when there are three users in each cluster for the
general setting. Without loss of generality, assume $M_1^1\ge M_2^1\ge M_3^1$
and $M_1^2\ge M_2^2\ge M_3^2$. The idea of the achievable DoF is similar to the
two-user case, and we thus only focus on identifying the optimal DoF
and describing the corresponding achievable schemes.
\begin{theorem}\label{thm-2cluster-3user}
For the two-cluster MIMO multi-way relay channel with three users in each
cluster where the users and the relay can have arbitrary number of
antennas, the optimal DoF is:

{\it (1) $DoF^* = 2N$ }
\begin{itemize}
\item When $N\le \max\{M_2^1,M_2^2\}$ or
\item When $N> \max\{M_2^1,M_2^2\}$ and one of the following
  conditions is satisfied:
  \begin{itemize}
    \item $M_1^1\ge N$, $M_1^2\ge N$.
    \item $M_1^1\ge N$, $M_1^2< N$,
      \begin{align}
        &M_2^1+M_3^1+(M_2^1+M_3^1-N)^++(M_1^2+M_2^2-N)^+\nonumber\\
        &+(M_1^2+M_3^2-N)^++(M_2^2+M_3^2-N)^+\ge N.
      \end{align}
    \item $M_1^1<N$, $M_1^2\ge N$,
      \begin{align}
        &(M_1^1+M_2^1-N)^++(M_1^1+M_3^1-N)^++(M_2^1+M_3^1-N)^+\nonumber\\
        &+M_2^2+M_3^2+(M_2^2+M_3^2-N)^+\ge N.
      \end{align}
    \item $M_1^1<N$, $M_1^2<N$,
      \begin{align}
        &(M_1^1+M_2^1-N)^++(M_1^1+M_3^1-N)^++(M_2^1+M_3^1-N)^+\nonumber\\
        &+(M_1^2+M_2^2-N)^++(M_1^2+M_3^2-N)^++(M_2^2+M_3^2-N)^+\ge N.
      \end{align}
  \end{itemize}
\end{itemize}

{\it (2) $DoF^*=M_1^1+M_2^1+M_3^1+M_1^2+M_2^2+M_3^2$}
\begin{itemize}
\item When {\it $N\ge
M_1^1+M_2^1+M_3^1+M_1^2+M_2^2+M_3^2$}.
\end{itemize}

{\it (3) $DoF^*=2(M_2^1+M_3^1+M_2^2+M_3^2)$}
\begin{itemize}
\item When $N\ge 2(M_2^1+M_3^1+M_2^2+M_3^2)$ or
\item When $N< 2(M_2^1+M_3^1+M_2^2+M_3^2)$ and one of the following
  conditions is satisfied:
  \begin{itemize}
    \item $M_1^1\ge M_2^1+M_3^1+M_2^2+M_3^2$, $M_1^2\ge
    M_2^1+M_3^1+M_2^2+M_3^2$.
    \item $N\ge 2(M_2^1+M_3^1)+M_2^2+M_3^2$, $M_1^2 \ge
    2(M_2^1+M_3^1)+M_2^2+M_3^2$.
    \item $N\ge M_2^1+M_3^1+2(M_2^2+M_3^2)$, $M_1^1 \ge
    M_2^1+M_3^1+2(M_2^2+M_3^2)$.
  \end{itemize}
\end{itemize}

{\it (4) $DoF^*=M_1^1+M_2^1+M_3^1+2(M_2^2+M_3^2)$}
\begin{itemize}
\item When $N\ge M_1^1+M_2^1+M_3^1+2(M_2^2+M_3^2)$ or
\item When $M_1^1+M_2^1+M_3^1+M_2^2+M_3^2\le N<
  M_1^1+M_2^1+M_3^1+2(M_2^2+M_3^2)$, $M_1^2\ge
  M_1^1+M_2^1+M_3^1+M_2^2+M_3^2$.
\end{itemize}

{\it (5) $DoF^*=2(M_2^1+M_3^1)+M_1^2+M_2^2+M_3^2$}
\begin{itemize}
\item When $N\ge 2(M_2^1+M_3^1)+M_1^2+M_2^2+M_3^2$ or
\item When $M_2^1+M_3^1+M_1^2+M_2^2+M_3^2\le
  N<2(M_2^1+M_3^1)+M_1^2+M_2^2+M_3^2$, $M_1^1\ge M_2^1+M_3^1+M_1^2+M_2^2+M_3^2$.
\end{itemize}
\end{theorem}

\begin{proof}
We first consider the DoF upperbound in {\it
  Theorem \ref{thm-DoF-upperbound}}. For the three user case, the upper bound reduces to
\begin{align}
DoF\le
&\Bigl\{2N,M_1^1+M_2^1+M_3^1+M_1^2+M_2^2+M_3^2,2(M_2^1+M_3^1+M_2^2+M_3^2),\\
&M_1^1+M_2^1+M_3^1+2(M_2^2+M_3^2),2(M_2^1+M_3^1)+M_1^2+M_2^2+M_3^2\Bigr\}.
\end{align}

We now investigate the following cases to establish the optimal DoF:

\subsubsection{When $2N$ is the binding term in the DoF upper bound} For
this case, we have
\begin{equation}
M_1^1+M_2^1+M_3^1+M_1^2+M_2^2+M_3^2\ge 2N, M_2^1+M_3^1+M_2^2+M_3^2\ge N.
\end{equation}

If we have $N\le \max\{M_2^1,M_2^2\}$, the DoF $2N$ can always be
achieved by only letting the two users with more antennas than the
relay to transmit using FDF schemes, which in fact reduces the channel
to a two-way relay channel.

If we have $N>\max\{M_2^1,M_2^2\}$, we have $N>M_2^1\ge M_3^1$ and
$N>M_2^2\ge M_3^2$. We consider the following scenarios:
\begin{itemize}
\item $M_1^1\ge N$ and $M_1^2\ge N$: From {\it Lemma \ref{lma}}, the
  dimension of shared signal space between the users is summarized in
  the table below:

\begin{center}
    \begin{tabular}{ | l | l || l | l |}
    \hline
    \multicolumn{2}{|c||}{Cluster 1} & \multicolumn{2}{c|}{Cluster 2}\\
    \hline
    User 1 and 2 & $M_2^1$ & User 1 and 2 & $M_2^2$ \\
    \hline
    User 1 and 3 & $M_3^1$ & User 1 and 2 & $M_3^2$ \\
    \hline
    User 2 and 3 & $(M_2^1+M_3^1-N)^+$ & User 2 and 3 &
    $(M_2^2+M_3^2-N)^+$ \\
    \hline
    \end{tabular}
\end{center}

Since we have $M_2^1+M_3^1+M_2^2+M_3^2\ge N$, The DoF upper bound $2N$
can thus be achieved by letting user 1 and user 2, and
user 1 and user 3 in each cluster to exchange messages using signal
space alignment, such that all the dimension of the signal space of
the relay is shared by one pair of users. The detailed scheme
is similar to the two-user case and is thus omitted.

\item $M_1^1\ge N$ and $N>M_1^2$: From {\it Lemma \ref{lma}}, the
  dimension of shared signal space between the users is summarized in
  the table below:

\begin{center}
    \begin{tabular}{ | l | l || l | l |}
    \hline
    \multicolumn{2}{|c||}{Cluster 1} & \multicolumn{2}{c|}{Cluster 2}\\
    \hline
    User 1 and 2 & $M_2^1$ & User 1 and 2 & $(M_1^2+M_2^2-N)^+$ \\
    \hline
    User 1 and 3 & $M_3^1$ & User 1 and 2 & $(M_1^2+M_3^2-N)^+$ \\
    \hline
    User 2 and 3 & $(M_2^1+M_3^1-N)^+$ & User 2 and 3 &
    $(M_2^2+M_3^2-N)^+$ \\
    \hline
    \end{tabular}
\end{center}

If we have
\begin{equation}
M_2^1+M_3^1+(M_2^1+M_3^1-N)^++(M_1^2+M_2^2-N)^++(M_1^2+M_3^2-N)^++(M_2^2+M_3^2-N)^+\ge N,
\end{equation}
then all the dimension of the signal space at the relay can be shared
by one pair of users. Using signal space alignment, user 1 and user 2 in cluster 1, user 1 and user
3 in cluster 1, and the rest user pairs with $M_i^l+M_k^l-N>0$ can
exchanges messages. The DoF upper bound $2N$ can thus be achieved.

\item $N>M_1^1$ and $M_1^2\ge N$: From {\it Lemma \ref{lma}}, the
  dimension of shared signal space between the users is summarized in
  the table below:

\begin{center}
    \begin{tabular}{ | l | l || l | l |}
    \hline
    \multicolumn{2}{|c||}{Cluster 1} & \multicolumn{2}{c|}{Cluster 2}\\
    \hline
    User 1 and 2 & $(M_1^1+M_2^1-N)^+$ & User 1 and 2 & $M_2^2$ \\
    \hline
    User 1 and 3 & $(M_1^1+M_3^1-N)^+$ & User 1 and 2 & $M_3^2$ \\
    \hline
    User 2 and 3 & $(M_2^1+M_3^1-N)^+$ & User 2 and 3 &
    $(M_2^2+M_3^2-N)^+$ \\
    \hline
    \end{tabular}
\end{center}

If we have
\begin{equation}
(M_1^1+M_2^1-N)^++(M_1^1+M_3^1-N)^++(M_2^1+M_3^1-N)^++M_2^2+M_3^2+(M_2^2+M_3^2-N)^+\ge N,
\end{equation}
then the DoF upper bound $2N$ can be achieved following similar
arguments as in the previous case.

\item $N>M_1^1$ and $N>M_1^2$: From {\it Lemma \ref{lma}}, the
  dimension of shared signal space between the users is summarized in
  the table below:

\begin{center}
    \begin{tabular}{ | l | l || l | l |}
    \hline
    \multicolumn{2}{|c||}{Cluster 1} & \multicolumn{2}{c|}{Cluster 2}\\
    \hline
    User 1 and 2 & $(M_1^1+M_2^1-N)^+$ & User 1 and 2 & $(M_1^2+M_2^2-N)^+$ \\
    \hline
    User 1 and 3 & $(M_1^1+M_3^1-N)^+$ & User 1 and 2 & $(M_1^2+M_3^2-N)^+$ \\
    \hline
    User 2 and 3 & $(M_2^1+M_3^1-N)^+$ & User 2 and 3 &
    $(M_2^2+M_3^2-N)^+$ \\
    \hline
    \end{tabular}
\end{center}

If we have
\begin{align}
&(M_1^1+M_2^1-N)^++(M_1^1+M_3^1-N)^++(M_2^1+M_3^1-N)^+\nonumber\\
&+(M_1^2+M_2^2-N)^++(M_1^2+M_3^2-N)^++(M_2^2+M_3^2-N)^+\ge N,
\end{align}
then the DoF upper bound $2N$ can be achieved by letting user pairs
with $M_i^l+M_j^l-N>0$ to exchange messages such that all dimension
of the signal space of the relay is utilized by a pair of users using
signal space alignment.
\end{itemize}

\subsubsection{When $M_1^1+M_2^1+M_3^1+M_1^2+M_2^2+M_3^2$ is the
  binding term in the DoF upper bound} For this case, we have
$M_1^1\le M_2^1+M_3^1$ and $M_1^2\le M_2^2+M_3^2$. The DoF
upper bound $M_1^1+M_2^1+M_3^1+M_1^2+M_2^2+M_3^2$ can be achieved if
$N\ge M_1^1+M_2^1+M_3^1+M_1^2+M_2^2+M_3^2$, i.e., users utilize the
relay in the multiple-access fashion and relay can decode all the messages from the
users and broadcast the messages to the intended users.

When $N< M_1^1+M_2^1+M_3^1+M_1^2+M_2^2+M_3^2$, it is easy to verify
that the dimension of the shared signal space between all the users is
always less than ${M_1^1+M_2^1+M_3^1+M_1^2+M_2^2+M_3^2 \over 2}$, and
whether the DoF upper bound can be achieved is unknown.

\subsubsection{When $2(M_2^1+M_3^1+M_2^2+M_3^2)$ is the binding term
  in the DoF upper bound}
For this case, we have $M_1^1>M_2^1+M_3^1$ and $M_1^2>M_2^2+M_3^2$. We
also have $N>M_2^1+M_3^1+M_2^2+M_3^2$, which means user 2 and user 3
in each cluster cannot share any dimension of the signal space of the
relay. The DoF upper bound can be achieved for the following cases:
\begin{itemize}
\item $N\ge 2(M_2^1+M_3^1+M_2^2+M_3^2)$: Under this condition, the DoF
  upper bound can be achieved by letting all the users use the relay
  in the multiple-access fashion. The relay can decode all the
  messages and then broadcast the messages back to the intended
  users.
\item $N< 2(M_2^1+M_3^1+M_2^2+M_3^2)$: Under this condition, the DoF
  upper bound can be achieve for the following cases:
  \begin{itemize}
  \item $M_1^1\ge M_2^1+M_3^1+M_2^2+M_3^2$ and $M_1^2\ge
    M_2^1+M_3^1+M_2^2+M_3^2$: The DoF upper bound can be achieved by
    only allowing the relay to use $M_2^1+M_3^1+M_2^2+M_3^2$ antennas
    to assist the users. Based on {\it Lemma \ref{lma}}, the dimension
    of the shared signal space between the users is
    \begin{center}
      \begin{tabular}{ | l | l || l | l |}
        \hline
        \multicolumn{2}{|c||}{Cluster 1} & \multicolumn{2}{c|}{Cluster 2}\\
        \hline
        User 1 and 2 & $M_2^1$ & User 1 and 2 & $M_2^2$ \\
        \hline
        User 1 and 3 & $M_3^1$ & User 1 and 3 & $M_3^2$ \\
        \hline
        User 2 and 3 & $0$ & User 2 and 3 & $0$ \\
        \hline
      \end{tabular}
    \end{center}
    The DoF upper bound can thus be achieved by letting user 1 and
    user 2, user 1 and user 3 in each cluster to exchange messages
    using signal space alignment.
  \item $N\ge 2(M_2^1+M_3^1)+M_2^2+M_3^2$ and $M_1^2 \ge
    2(M_2^1+M_3^1)+M_2^2+M_3^2$: The DoF upper bound can be achieved
    by only allowing the relay to use $2(M_2^1+M_3^1)+M_2^2+M_3^2$
    antennas to assist the users. Based on {\it Lemma \ref{lma}}, user
    1 and user 2 in cluster 2 share $M_2^2$ dimensional signal space,
    and user 1 and user 3 in cluster 2 share $M_3^2$ dimensional
    signal space. These pairs of users occupy $M_2^2+M_3^2$
    dimensional signal space at the relay, and can be used to exchange
    $2(M_2^2+M_3^2)$ messages. Users in
    cluster 1 can utilize the rest $2(M_2^1+M_3^1)$ dimensional signal
    space in the multiple-access fashion to exchange $2(M_2^1+M_3^1)$
    messages. The DoF upper bound $2(M_2^1+M_3^1+M_2^2+M_3^2)$ can
    thus be achieved.
  \item $N\ge M_2^1+M_3^1+2(M_2^2+M_3^2)$ and $M_1^1 \ge
    M_2^1+M_3^1+2(M_2^2+M_3^2)$: The DoF upper bound $2(M_2^1+M_3^1+M_2^2+M_3^2)$ can
    be achieved by only allowing the relay to use
    $M_2^1+M_3^1+2(M_2^2+M_3^2)$ antennas, following similar arguments
    as in the previous case.
  \end{itemize}
\end{itemize}

\subsubsection{When $M_1^1+M_2^1+M_3^1+2(M_2^2+M_3^2)$ is the binding term
  in the DoF upper bound} \label{sec:DoFupperBound4}
For this case we have $M_1^1<M_2^1+M_3^1$
and $M_1^2>M_2^2+M_3^2$. We also have
\begin{equation}
N>\frac{M_1^1+M_2^1+M_3^1}{2}+M_2^2+M_3^2.
\end{equation}

The DoF upper bound can be achieved for the
following cases:
\begin{itemize}
\item $N\ge M_1^1+M_2^1+M_3^1+2(M_2^2+M_3^2)$: The DoF upper bound can
  be simply achieved by letting the users exchange their messages
  using the relay in the multiple-access fashion. The relay can decode
  all the messages and broadcast the messages back to the intended
  users since it has sufficient spatial dimension.
\item $M_1^1+M_2^1+M_3^1+M_2^2+M_3^2\le N<
  M_1^1+M_2^1+M_3^1+2(M_2^2+M_3^2)$ and $M_1^2\ge
  M_1^1+M_2^1+M_3^1+M_2^2+M_3^2$: The DoF upper bound can be achieved
  by only allowing the relay to use $M_1^1+M_2^1+M_3^1+M_2^2+M_3^2$
  antennas to assist the users. Based on {\it Lemma \ref{lma}}, user 1
  and user 2 in cluster 2 share $M_2^2$ dimensional space, and user 1
  and user 3 in cluster 2 share $M_3^2$ dimensional space at the
  relay, which allows the users to exchange $2(M_2^2+M_3^2)$
  messages using $M_2^2+M_3^2$ dimensional space at the relay. The
  rest $M_1^1+M_2^1+M_3^1$ dimensional signal space at the relay can
  be used to assist users in cluster 1 to exchange
  $M_1^1+M_2^1+M_3^1$ messages in the multiple-access
  fashion. Therefore the DoF $M_1^1+M_2^1+M_3^1+2(M_2^2+M_3^2)$ can be
  achieved.
\end{itemize}

\subsubsection{When $2(M_2^1+M_3^1)+M_1^2+M_2^2+M_3^2$ is the binding term
  in the DoF upper bound} For this case, the DoF upper bound can be achieved for
the following scenarios:
\begin{itemize}
\item $N\ge 2(M_2^1+M_3^1)+M_1^2+M_2^2+M_3^2$.
\item $M_2^1+M_3^1+M_1^2+M_2^2+M_3^2\le
  N<2(M_2^1+M_3^1)+M_1^2+M_2^2+M_3^2$ and $M_1^1\ge M_2^1+M_3^1+M_1^2+M_2^2+M_3^2$.
\end{itemize}
This case is similar to case 4), and the details are thus omitted.
\end{proof}

\section{$L$-Cluster $K$-User MIMO multi-way Relay Channel}\label{sec:general-multiway}
Consider now the general $L$-cluster $K$-user
MIMO multi-way relay channel in the symmetric setting, i.e., all the users
have the same number of antennas. We have the following optimal DoF result.

\begin{theorem}\label{thm-Lcluster-Kuser}
For the symmetric $L$-cluster $K$-user MIMO multi-way relay channel, where
all users have $M$ antennas and the relay has $N$ antennas, the
optimal DoF is
\begin{equation}
DoF^*=KLM\quad {\rm if}~N\ge KLM,
\end{equation}
\begin{equation}
DoF^*=2N\quad {\rm if}~\frac{LK(K-1)}{2}(2M-N)\ge N.
\end{equation}
\end{theorem}

To establish the optimal DoF, we first study the DoF upperbound. For
this case, the DoF upperbound in equation
(\ref{eq:4}) becomes
\begin{equation}
DoF\le \min\left\{KLM,2N\right\}.
\end{equation}

To investigate the achievability of the DoF upperbound, we further
consider the following cases:

\subsection{Achieving DoF $KLM$: Multiple-access transmission}
When $2N > KLM$, the DoF upper bound becomes $KLM$. The DoF upper
bound can be achieved when $N\ge KLM$. Under this
condition, the relay can decode all the messages from all the users
and can broadcast the messages to the intended users without inducing any
interference.

\subsection{Achieving DoF $2N$: Signal Space Alignment}
  When $2N\le KLM$, the DoF upper bound becomes $2N$. To achieve
  this upperbound, we require each signal dimension at the relay to be
  shared by a pair of users. From {\it Lemma \ref{lma}}, any pair of
  users in the same cluster can share $2M-N$ dimensional signal space
  at the relay, if $2M\ge N$. Therefore, we need
  \begin{equation}
   L{K \choose 2}(2M-N)\ge N,
  \end{equation} or equivalently
  \begin{equation}
   \frac{LK(K-1)}{2}(2M-N)\ge N,
  \end{equation} such that all the signal dimension at the relay can
  be shared by a pair of users. We can choose any pair of users to
  exchange data streams without exceeding their maximum allowed
  dimension of shared signal space $2M-N$. We let the users exchange
  $N$ pairs of data streams, and the relay can decode the sum of each pair of data
  stream and broadcast to the users with proper receiver-side processing. The
  detailed transmission scheme is described as follows:

  If $n=\frac{2N}{LK(K-1)}$ is an integer, we let user $i$ and
  user $j$, $i,j=1,2,\cdots,K$ in cluster $l$, $l=1,2,\cdots,L$
  exchange $n$ data streams, each with unit DoF. Since we have
  $2M-N\ge n$, based on {\it Lemma \ref{lma}}, each pair of users can transmit the data streams along
  the vectors ${\bf v}_{(ij),m}^l$ and ${\bf v}_{(ji),m}^l$, $m=1,2,\cdots,n$ such that
  \begin{equation}
   {\bf H}_{R(i,l)}{\bf v}_{(ij),m}^l={\bf H}_{R(j,l)}{\bf v}_{(ji),m}^l={\bf q}_{(ij),m}^l,
  \end{equation} where ${\bf v}_{(ij),m}^l$ denotes the $m$th beamforming
  vector for user $i$ in cluster $l$ to share the signal space at the relay with
  user $j$ in cluster $l$.

  The transmitted signal from user $i$ in cluster $l$ is thus
  \begin{equation}
   {\bf X}_{i,l}=\sum_{\substack{j=1\\j\neq i}}^K\sum_{m=1}^n{\bf v}_{(ij),m}^ld_{(ij),m}^l,
  \end{equation} where $d_{(ij),m}^l$ denotes the $m$th message from
  user $i$ in cluster $l$ for user $j$ in cluster $l$.

  The received signal at the relay is
  \begin{align}
    {\bf Y}_R&=\sum_{l=1}^L\sum_{i=1}^K{\bf H}_{R(i,l)}{\bf X}_{i,l}\\
    &=\sum_{l=1}^L\sum_{i=1}^K\sum_{\substack{j=1\\j\neq
        i}}^K\sum_{m=1}^n{\bf H}_{R(i,l)}{\bf
      v}_{(ij),m}^ld_{(ij),m}^l\\
    &=\sum_{l=1}^L\sum_{i=1}^K\sum_{\substack{j=i+1\\j\neq
        i}}^K\sum_{m=1}^n{\bf q}_{(ij),m}^l(d_{(ij),m}^l+d_{(ji),m}^l).
  \end{align}

  The relay can decode $d_{(ij),m}^l+d_{(ji),m}^l$ and then need to
  broadcast the messages back to the users. Following a similar scheme
  as in the two-cluster case, we let user $i$ and user $j$ employ a receiver-side
  filter $\left({\bf u}_{(ij),m}^l\right)$ and $\left({\bf
    u}_{(ji),m}^l\right)$ to decode the message
  $d_{(ij),m}^l+d_{(ji),m}^l$, where
  \begin{equation}
    \left({\bf u}_{(ij),m}^l\right)^T{\bf H}_{(i,l)R}=\left({\bf
      u}_{(ji),m}^l\right)^T{\bf H}_{R(j,l)}={\bf g}_{(ij),m}^l.\label{eq:13}
  \end{equation}

  Finding the receiver-side filter is a dual problem to finding the
  beamforming vector ${\bf v}_{(ij),m}^l$, which can be seen by taking transpose of
  equation (\ref{eq:13}). Based on {\it Lemma
    \ref{lma}}, there exist $2M-N$ such pair of vectors $\left({\bf u}_{(ij),m}^l\right)^T$ and $\left({\bf
    u}_{(ji),m}^l\right)^T$. User $i$ in cluster $l$ can choose $n$ out of these
  vectors to form a filtering matrix ${\bf U}_{(ij)}^l$ to receive the
  messages $d_{(ij),m}^l+d_{(ji),m}^l$, where the matrix ${\bf U}_{(ij)}^l$ is formed by taking
  $\left({\bf u}_{(ij),m}^l\right)^T$ as its rows. We can also combine
  the matrices ${\bf U}_{(ij)}^l$ for all $j=1,\cdots,K$, $j\neq
  i$ to form the filtering matrix for user $i$ to decode all the
  intended messages:
  \begin{equation}
    {\bf U}_{i}^l=\left[\begin{array}{c}
                        {\bf U}_{(i1)}^l\\
                        \vdots\\
                        {\bf U}_{(i,i-1)}^l\\
                        {\bf U}_{(i,i+1)}^l\\
                        \vdots\\
                        {\bf U}_{(iK)}^l\\
    \end{array}\right].
  \end{equation}

  The relay can use zero-forcing to broadcast the messages
  $d_{(ij),m}^l+d_{(ji),m}^l$ to the intended users. The users can
  decode the intended messages using their side information. The
  DoF $2N$ is thus achievable.

  When $n=\frac{2N}{LK(K-1)}$ is not an integer, we can let one pair
  of users to exchange
  \begin{equation}
    N-\left(\frac{LK(K-1)}{2}-1\right)\left\lceil\frac{2N}{LK(K-1)}\right\rceil
  \end{equation} messages, and the other user pairs exchange
  \begin{equation}
    \left\lceil\frac{2N}{LK(K-1)}\right\rceil
  \end{equation} messages, and the DoF $2N$ is still achievable.

\begin{remark}
{\it Theorem \ref{thm-Lcluster-Kuser}} provides us with the first DoF result for the $L$-cluster, $K$-user MIMO multi-way relay
channel for arbitrary $L$, $K$. We can see that the DoF is always limited by the available
spatial dimension at the relay, and that with fixed number of antennas at the
relay, increasing the number of users and the number of clusters
cannot provide DoF gain. In addition, we gain the insight that the DoF optimal way to utilize the resources of the relay is to
share the relay between two users. We cannot obtain DoF gain by
letting three or more users to share the resources of the relay.
\end{remark}
\begin{remark}
The result for the asymmetric case of the general $L$-cluster $K$-user
MIMO multi-way relay channel can be obtained following similar arguments as
in the two-cluster case. Other than having to enumerate a number of cases and conditions, the results do not provide further insights to what we already provide for the symmetric case. Hence the detailed expressions for the asymmetric cases are omitted here.
\end{remark}

\section{Conclusion}\label{sec:conclusion}
In this paper, we have investigated the DoF for the general MIMO multi-way
relay channel and established the optimal DoF for several scenarios of
interests. We have derived a new DoF upper bound using genie-aided
approach, which is shown to be tight for
several scenarios of interests. Specifically, we have studied the DoF for the two-cluster two-user
MIMO multi-way relay channel and two-cluster three-user MIMO multi-way relay
channel with arbitrary number of antennas, and established the optimal
DoF using signal space alignment, multiple-access transmission, or a combination of
both, depending on the number of antennas at the
users and the relay. We have also studied the $L$-cluster $K$-user MIMO multi-way relay
channel with equal number of antennas at the users, and established
the optimal DoF. The DoF results imply that the DoF of the
MIMO multi-way relay channel is always limited by the spatial dimension
available at the relay. With fixed number of antennas at the relay,
increasing the number of users and clusters cannot provide any DoF
gain. The results also imply that allowing three or more
users to share the resources of the relay cannot provide any DoF gain. 

This work has established the optimal DoF for a variety of scenarios for the multi-way relay channel which was unknown previously. For the remaining cases, determining the strategies to achieve the optimal DoF remains open.

\appendices
\section{Proof of Lemma \ref{lma}}\label{sec:proof-lemma-reflma}
\begin{proof}
We first consider the case when $p\ge q_1 \ge q_2$ and
$q_1+q_2>p$. Note that equation (\ref{eq:1}) is equivalent as
\begin{equation}
\left[\begin{array}{ccc}
         {\bf I} & {\bf H}_{1} & {\bf 0}\\
         {\bf I} & {\bf 0} & {\bf H}_{2}\\
\end{array}\right]
\left[\begin{array}{c}
         {\bf v}_{i}\\
         {\bf u}_{i}\\
         {\bf w}_{i}\\
      \end{array}
\right]={\bf 0} \label{eq:14}.
\end{equation}

The null space of the matrix
\begin{equation}
\left[\begin{array}{ccc}
         {\bf I} & {\bf H}_{1} & {\bf 0}\\
         {\bf I} & {\bf 0} & {\bf H}_{2}\\
\end{array}\right]\label{eq:15}
\end{equation} has dimension $q_1+q_2-p$. It is easy to see that if
$q_1+q_2>p$, then we can find $q_1+q_2-p$ non-zero linearly independent vectors
of the form
\begin{equation}
\left[\begin{array}{ccc}
         {\bf v}_{i} & {\bf u}_{i} & {\bf w}_{i}\\
      \end{array}
\right]^T \label{eq:17}
\end{equation}
from the null space of the matrix shown in equation (\ref{eq:15}). It
remains to see whether all these vectors satisfy ${\bf v}_{i}\ne
0$. Since $p\ge q_1\ge q_2$, we can see that the null space of matrices ${\bf
  H}_{1}$ and ${\bf H}_{2}$ has dimension 0. Therefore for all the
non-zero vectors satisfying equation (\ref{eq:14}), we must have ${\bf v}_{i}\ne
0$.

Similarly, when $q_1\ge p \ge q_2$, we can find $q_1+q_2-p$ non-zero
linearly independent vectors of the form shown in equation
(\ref{eq:17}) to satisfy equation (\ref{eq:14}). However, for this
case, if we consider the equation ${\bf v}_i={\bf H}_2{\bf w}_{i}$, we
can see that there are at most $q_2$ non-zero linearly independent
vectors ${\bf v}_i$ satisfying this equation. In fact, since $q_1\ge p$, the null space of matrix ${\bf H}_{1}$ has
dimension $q_1-p$. When we set ${\bf w}_i$ and ${\bf v}_i$ to 0, we can find $q_1-p$
non-zero linearly independent vectors ${\bf u}_i$ to satisfy equation
(\ref{eq:14}). Therefore we can conclude that among all vectors of the
form in equation (\ref{eq:17}) satisfying equation (\ref{eq:14}),
we can only find $q_2$ non-zero linearly independent
vectors ${\bf v}_i$.
\end{proof}

\bibliographystyle{unsrt}

\begin{thebibliography}{10}
\bibitem{Multiway}
D.~G\"und\"uz, A.~Yener, A.~J. Goldsmith, and H.~V. Poor, ``The multi-way relay
  channel,'' \emph{IEEE Transactions on Information Theory}, vol.~59, no.~1,
  pp. 51--63, Jan 2013.

\bibitem{shannontwoway}
C.~E. Shannon, ``Two-way communication channels,'' in \emph{Proceedings of 4th
  Berkeley Symposium on Math, Statistics and Probability}, 1961, pp. 611--644.

\bibitem{Aves_twoway}
A.~S. Avestimehr, A.~Sezgin, and D.~Tse, ``Capacity of the two-way relay
  channel within a constant gap,'' \emph{European Transactions on
  Telecommunications}, 2009, dOI: 10.1002/ett.000.

\bibitem{twowayconst}
W.~Nam, S.~Y. Chung, and Y.~H. Lee, ``Capacity of the {G}aussian two-way relay
  channel to within 1/2 bit,'' \emph{IEEE Transactions on Information Theory},
  vol.~56, no.~11, pp. 5488 -- 5494, November 2011.

\bibitem{bidirection}
M.~P. Wilson, K.~Narayanan, H.~Pfister, and A.~Sprintson, ``Joint physical
  layer coding and network coding for bi-directional relaying,'' \emph{IEEE
  Transactions on Information Theory}, vol.~56, no.~11, November 2010.

\bibitem{MinChenMUTWRC}
M.~Chen and A.~Yener, ``Power allocation for {F/TDMA} multiuser two-way relay
  networks,'' \emph{IEEE Transactions on Wireless Communications}, vol.~9,
  no.~2, pp. 546--551, Febuary 2010.

\bibitem{Ong}
L.~Ong, S.~J. Johnson, and C.~M. Kellett, ``The capacity region of multiway
  relay channels over finite fields with full data exchange,'' \emph{IEEE
  Transactions on Information Theory}, vol.~57, no.~5, pp. 3016--3031, May
  2011.

\bibitem{Equalrate_twoway}
L.~Ong, C.~M. Kellett, and S.~J. Johnson, ``On the equal-rate capacity of the
  {AWGN} multiway relay channel,'' \emph{IEEE Transactions on Information
  Theory}, vol.~58, no.~9, pp. 5761 -- 5769, September 2012.

\bibitem{LinearY}
A.~Chaaban and A.~Sezgin, ``The capacity region of the linear shift
  deterministic {Y}-channel,'' in \emph{Proceedings of IEEE International
  Symposium on Information Theory}, July 2011.

\bibitem{Ysum}
A.~Chaaban, A.~Sezgin, and A.~S. Avestimehr, ``On the sum capacity of the
  {Y}-channel,'' in \emph{Proceedings of the 45th Asilomar Conference on
  Signals, Systems and Computers}, November 2011.

\bibitem{divide}
A.~Sezgin, A.~S. Avestimehr, M.~A. Khajehnejad, and B.~Hassibi,
  ``Divide-and-conquer: Approaching the capacity of the two-pair bidirectional
  {G}aussian relay network,'' \emph{IEEE Transactions on Information Theory},
  vol.~58, no.~4, pp. 2434 -- 2454, April 2012.

\bibitem{minchen_mutwoway}
M.~Chen and A.~Yener, ``Multiuser two-way relaying: Detection and interference
  management strategies,'' \emph{IEEE Transactions on Wireless Communications},
  vol.~8, no.~8, pp. 4296--4305, August 2009.

\bibitem{DoFMIMOX}
S.~A. Jafar and S.~Shamai, ``Degrees of freedom region for the {MIMO X}
  channel,'' \emph{IEEE Transactions on Information Theory}, vol.~54, no.~1,
  pp. 151--170, January 2008.

\bibitem{CadambeKuser}
V.~R. Cadambe and S.~A. Jafar, ``Interference alignment and the degrees of
  freedom for the {K}-user interference channel,'' \emph{IEEE Transactions on
  Information Theory}, vol.~54, no.~8, pp. 3425--3441, August 2008.

\bibitem{CadambeXNetwork}
------, ``Interference alignment and the degrees of freedom of wireless {X}
  networks,'' \emph{IEEE Transactions on Information Theory}, vol.~55, no.~9,
  pp. 3893--3908, September 2009.

\bibitem{TiangaoKMN}
T.~Gou and S.~A. Jafar, ``Degrees of freedom of the {K} user {MxN MIMO}
  interference channel,'' \emph{IEEE Transactions on Information Theory},
  vol.~56, no.~12, pp. 6040--6057, December 2010.

\bibitem{signalspaceY}
N.~Lee, J.-B. Lim, and J.~Chun, ``Degrees of freedom of the {MIMO} {Y} channel:
  signal space alignment for network coding,'' \emph{IEEE Transactions on
  Information Theory}, vol.~56, no.~7, pp. 3332--3342, July 2010.

\bibitem{KYchannel}
K.~Lee, N.~Lee, and I.~Lee, ``Achievable degrees of freedom on {K}-user {Y}
  channels,'' \emph{IEEE Transactions on Wireless Communications}, vol.~11,
  no.~3, pp. 1210 -- 1219, March 2012.

\bibitem{IAinMUTWRC}
R.~S. Ganesan, T.~Weber, and A.~Klein, ``Interference alignment in multi-user
  two way relay networks,'' in \emph{Proceedings of 2011 IEEE 73rd Vehicular
  Technology Conference}, May 2011.

\bibitem{DoF-MIMO_ML_TWRC}
K.~Lee, S.-H. Park, J.-S. Kim, and I.~Lee, ``Degrees of freedom on {MIMO}
  multi-link two-way relay channels,'' in \emph{Proceedings of 2010 IEEE Global
  Telecommunication Conference}, Dec 2010.

\bibitem{asymmw}
F.~Sun and E.~de~Carvalho, ``Degrees of freedom of asymmetrical multi-way relay
  network,'' in \emph{Proceedings of 2011 IEEE 12th International Workshop on
  Signal Processing Advances in Wireless Communications}, June 2011.

\bibitem{YeISIT13}
Y.~Tian and A.~Yener, ``Degrees of freedom for the mimo multi-way relay
  channel,'' in \emph{Proceedings of IEEE International Symposium on
  Information Theory}, July 2013.

\bibitem{Uri_AWGN}
U.~Erez and R.~Zamir, ``Achieving $1/2\log(1+ {SNR})$ on the {AWGN} channel
  with lattice encoding and decoding,'' \emph{IEEE Transactions on Information
  Theory}, vol.~50, no.~10, pp. 2293--2314, October 2004.
\end{thebibliography}

\end{document}